\documentclass[11pt]{article}
\usepackage{amsmath,amsfonts,amsthm,amssymb,color}
\usepackage{fancybox}
\usepackage{pdfsync}
\usepackage{fullpage}
\usepackage{hyperref}
\usepackage{algorithm} 
\usepackage[noend]{algpseudocode}
\usepackage{algorithmicx}
\usepackage{thmtools}
\usepackage{thm-restate}

\newenvironment{fminipage}%
  {\begin{Sbox}\begin{minipage}}%
  {\end{minipage}\end{Sbox}\fbox{\TheSbox}}

\def\abs#1{\left|#1  \right|}

\def\norm#1{\left\| #1 \right\|}

\newcommand\Dhat{{\widehat{{\Delta}}}}
\newcommand\Dtil{{\widetilde{{\Delta}}}}
\newcommand\Dbar{{\bar{{\Delta}}}}
\newcommand\Dopt{{{{\Delta^{\star}}}}}
\newcommand\residual{{{\sc res}}}

\DeclareMathOperator{\poly}{poly}

\def\showauthornotes{1}

\def\showdraftbox{0}

\usepackage{hyperref}
\usepackage{amsmath,amssymb,amsthm,amstext,amsfonts,bbm,graphicx,xspace,nicefrac}
\usepackage{color,stmaryrd,enumerate,latexsym,bm,amsfonts,
  subfigure,wrapfig,verbatim, tabularx,textcomp}
\usepackage[small]{caption}
\usepackage{comment} 
\usepackage{epsfig} 
\usepackage{latexsym,nicefrac,bbm}
\usepackage{xspace}
\usepackage{color,fancybox,graphicx,url,subfigure}
\usepackage{enumitem, fullpage}
\usepackage{booktabs}
\usepackage{commath}
\usepackage{mdframed}

\newtheorem{theorem}{Theorem}[section]

\newtheorem{lemma}[theorem]{Lemma}
\newtheorem{corollary}[theorem]{Corollary}

\theoremstyle{definition}
\newtheorem{definition}[theorem]{Definition}

\newcommand{\nfrac}[2]{\nicefrac{#1}{#2}}
\def\abs#1{\left| #1 \right|}
\renewcommand{\norm}[1]{\ensuremath{\left\lVert #1 \right\rVert}}

\newcommand\rea{\mathbb R}

\newcommand{\marginlabel}[1]%
{\mbox{}\marginpar{\it{\raggedleft\hspace{0pt}#1}}}

\definecolor{Mygray}{gray}{0.8}

 \ifcsname ifcommentflag\endcsname\else
  \expandafter\let\csname ifcommentflag\expandafter\endcsname
                  \csname iffalse\endcsname
\fi

\ifnum\showauthornotes=1
\newcommand{\todo}[1]{\colorbox{Mygray}{\color{red}#1}}
\else
\newcommand{\todo}[1]{}
\fi

\ifnum\showauthornotes=1
\newcommand{\Authornote}[2]{{\sf\small\color{red}{[#1: #2]}}}
\newcommand{\Authoredit}[2]{{\sf\small\color{red}{[#1]}\color{blue}{#2}}}
\newcommand{\Authorcomment}[2]{{\sf \small\color{gray}{[#1: #2]}}}
\newcommand{\Authorfnote}[2]{\footnote{\color{red}{#1: #2}}}
\newcommand{\Authorfixme}[1]{\Authornote{#1}{\textbf{??}}}
\newcommand{\Authormarginmark}[1]{\marginpar{\textcolor{red}{\fbox{
#1:!}}}}
\else
\newcommand{\Authornote}[2]{}
\newcommand{\Authoredit}[2]{}
\newcommand{\Authorcomment}[2]{}
\newcommand{\Authorfnote}[2]{}
\newcommand{\Authorfixme}[1]{}
\newcommand{\Authormarginmark}[1]{}
\fi

\newlength{\pgmtab}  
\let\originalleft\left
\let\originalright\right
\renewcommand{\left}{\mathopen{}\mathclose\bgroup\originalleft}
  \renewcommand{\right}{\aftergroup\egroup\originalright}

\def\abs#1{\left|#1  \right|}

\def\norm#1{\left\| #1 \right\|}

\newcommand\bb{\boldsymbol{\mathit{b}}}

\newcommand\dd{\boldsymbol{\mathit{d}}}

\newcommand\ff{\boldsymbol{\mathit{f}}}

\renewcommand\gg{\boldsymbol{\mathit{g}}}

\newcommand\rr{\boldsymbol{\mathit{r}}}

\def\tt{\boldsymbol{\mathit{t}}}

\newcommand\ww{\boldsymbol{\mathit{w}}}

\newcommand\xx{\boldsymbol{\mathit{x}}}

\renewcommand\AA{\boldsymbol{\mathit{A}}}
\newcommand\BB{\boldsymbol{\mathit{B}}}

\newcommand\Otil{\widetilde{O}}

\newcommand{\opt}{\textsc{Opt}}

\newcommand{\vone}{\boldsymbol{\mathbf{1}}}
\newcommand{\vzero}{\boldsymbol{\mathbf{0}}}

\newcommand{\etal}{\emph{et al.}}

 {
	\begin{enumerate}}{\end{enumerate}}

\def\qedsketch{\ifmmode\Box\else{\unskip\nobreak\hfil
\penalty50\hskip1em\null\nobreak\hfil$\Box$
\parfillskip=0pt\finalhyphendemerits=0\endgraf}\fi}

\newlength{\tpush}
\newcommand{\handout}[5]{
   \noindent
   \begin{center}
   \framebox{ \vbox{ \hbox to \textwidth { {\bf \coursenum\ :\  \coursename} \hfill #5 }
       \vspace{3mm}
       \hbox to \textwidth { {\Large \hfill #2  \hfill} }
       \vspace{1mm}
       \hbox to \textwidth { {\it #3 \hfill #4} }
     }
   }
   \end{center}
   \vspace*{4mm}
   \newcommand{\lecturenum}{#1}
   \addcontentsline{toc}{chapter}{Lecture #1 -- #2}
}

\ifnum\showdraftbox=1

\else

\fi

\allowdisplaybreaks

\usepackage[
  backend=biber,
  backref=true,
  backrefstyle=none,
  date=year,
  doi=false,
  giveninits=true,
  hyperref=true,
  maxbibnames=10,
  sortcites=false,
  style=alphabetic,
  url=false, 
]{biblatex}
\newcommand{\cost}{\texttt{cost}}

\newcommand{\eps}{\varepsilon}

\begin{document}

\title{
  Faster $p$-norm minimizing flows, via smoothed $q$-norm problems}
\author{Deeksha Adil\\Department of Computer Science\\University of Toronto\\deeksha@cs.toronto.edu \and Sushant Sachdeva\\Department of Computer Science\\University of Toronto\\sachdeva@cs.toronto.edu}

\maketitle
\begin{abstract}

  We present faster high-accuracy algorithms for computing
  $\ell_p$-norm minimizing flows. On a graph with $m$ edges, our
  algorithm can compute a $(1+1/\text{poly}(m))$-approximate
  unweighted $\ell_p$-norm minimizing flow with
  $pm^{1+\frac{1}{p-1}+o(1)}$ operations, for any $p \ge 2,$ giving
  the best bound for all $p\gtrsim 5.24.$ Combined with the algorithm
  from the work of Adil et al. (SODA '19), we can now compute such
  flows for any $2\le p\le m^{o(1)}$ in time at most $O(m^{1.24}).$ In
  comparison, the previous best running time was $\Omega(m^{1.33})$
  for large constant $p.$ For $p\sim\delta^{-1}\log m,$ our algorithm
  computes a $(1+\delta)$-approximate maximum flow on undirected
  graphs using $m^{1+o(1)}\delta^{-1}$ operations, matching the
  current best bound, albeit only for unit-capacity graphs.

  We also give an algorithm for solving general $\ell_{p}$-norm
  regression problems for large $p.$ Our algorithm makes
  $pm^{\frac{1}{3}+o(1)}\log^2(1/\varepsilon)$ calls to a
  linear solver. This gives the first high-accuracy algorithm for
  computing weighted $\ell_{p}$-norm minimizing flows that runs in
  time $o(m^{1.5})$ for some $p=m^{\Omega(1)}.$

  Our key technical contribution is to show that smoothed
  $\ell_p$-norm problems introduced by Adil et al., are interreducible
  for different values of $p.$
  No such reduction is known for standard $\ell_p$-norm problems.

\end{abstract}
\thispagestyle{empty}
\newpage
\setcounter{page}{1}
\section{Introduction}
Network flow problems are some of the most well-studied problems in
algorithms and combinatorial optimization
(e.g. see~\cite{AhujaMO93, Schrijver02, GoldbergT14}).
A generic
network flow problem can be expressed as follows: given a graph
$G(V,E)$ with $n$ vertices and $m$ edges, and a vector
$\bb \in \rea^{V}$ specifying net demands on vertices, we seek to find
a flow $\ff \in \rea^{E}$ that satisfies the demands and minimizes
some specified cost function of $\ff$
\[\min_{\BB^{\top} \ff = \dd} \cost(\ff).\]
Here $\BB$ is the edge-vertex incident matrix of the graph $G,$
\emph{i.e.} the column $\BB_{(u,v)}$ has a $+1$ entry at $u,$ $-1$ at
$v,$ and 0 otherwise.

Different choices of the $\cost$ function in the above formulation
capture various extensively-studied questions; a weighted
$\ell_1$-norm yields the classic shortest paths problem (or more
generally, transshipment), whereas a weighted $\ell_\infty$-norm
yields the extensively-studied maximum-flow problem (as min-congestion
flows) (see~\cite{Madry13} for a survey of its extensive history).

Picking $\cost$ as a weighted $\ell_2$ norm yields electrical
flows. The seminal work of Spielman and Teng~\cite{SpielmanT04} gave
the first nearly-linear time algorithm for solving the weighted
$\ell_2$ version to high-accuracy (a $(1+\eps)$-approximate solution
in time $\Otil(m\cdot\log \frac{1}{\eps})$ \footnote{The
  $\Otil(\cdot)$ notation hides $\poly(\log mp)$ factors.}). The work
of Spielman-Teng and the several followup-works have led to the
fastest algorithms for maximum matching~\cite{Madry13}, shortest paths
with negative weights~\cite{CohenMSV17}, graph
partitioning~\cite{OrecchiaSV12}, sampling random spanning
trees~\cite{KelnerM09, MadryST15, Schild18}, matrix
scaling~\cite{CohenMTV17, ZhuLOW17}, and resulted in dramatic progress
on the problem of computing maximum flows.

The work of Spielman and Teng inspired an exciting sequence of works
on the maximum flow problem~\cite{ChristianoKMST10, Sherman13,
  KelnerLOS14, Peng16, Sherman17, SidfordT18} that combines
combinatorial ideas with convex optimization, and has resulted in
nearly-linear time algorithms for approximating maximum flow on
undirected graphs, obtaining a $(1+\eps)$-approximation in time
$\Otil(m \eps^{-1}).$ Note that all these results are \emph{low
  accuracy} in the sense that the dependence of the running time on
$\eps$ is $\poly(\eps^{-1}).$ In contrast, a high-accuracy
algorithm, such as the Spielman-Teng algorithm, has a
$\poly(\log \nfrac{1}{\eps})$ dependence, and hence allows us to
compute a $(1+\frac{1}{\poly(n)})$ approximation with only an
$\Otil(1)$ factor loss in running time. A high-accuracy almost-linear
time algorithm for undirected maximum-flow would imply an
almost-linear time algorithm for exact maximum flow on directed graphs
(with polynomially bounded capacities). Such an algorithm remains an
outstanding open problem in the area, and the current best algorithms
run in time $\Otil(m\sqrt{n})$~\cite{LeeS14} and
$\Otil(m^{\nfrac{10}{7}})$ for unit capacity graphs~\cite{Madry13}.

In this paper, we study an important intermediate class of problems
interpolating between the electrical-flow problem ($\ell_2$) and the
maximum-flow problem ($\ell_{\infty}$) case, obtained by selecting
$\cost$ as an $\ell_{p}$-norm for some $p \in (2,\infty).$ Given that
we have high-accuracy algorithms for the $p=2$ case, one might
reasonably expect it to be easier to compute a
$(1+\frac{1}{\poly(n)})$-approximate $\ell_p$-norm minimizing flows
compared to max-flow ($p=\infty$). However, prior to 2018, the best
algorithm for this problem was obtained by combining almost-linear
time electrical-flow algorithms~\cite{SpielmanT04, CohenKMPPRX14} with
interior point methods~\cite{NesterovN94}, resulting in a running time
of $\Otil(m^{\nfrac{3}{2}})$ for all $p \in (2,\infty).$ The work of
Bubeck~\etal.~\cite{BubeckCLL18} gave a new homotopy-approach for
$\ell_{p}$-regression problems, giving a running time of
$\Otil(\poly(p)m^{\frac{p-2}{2p} + 1}) \le \Otil(
\poly(p)m^{\nfrac{3}{2}}).$ Adil~\etal.~\cite{AdilKPS19} gave an
iterative refinement scheme for $\ell_p$-norm regression, giving a
running time of
$\Otil(p^{O(p)}{\cdot} m^{\frac{p-2}{3p-2} + 1} ) \le \Otil(
p^{O(p)}{\cdot} m^{\nfrac{4}{3}}).$ Building on the work of
Adil~\etal, Kyng~\etal.~\cite{KyngPSW19} designed an algorithm for
unweighted $p$-norm flows that runs in time
$\exp(p^{\nfrac{3}{2}}) m^{1 + \frac{7}{\sqrt{p-1}} + o(1)}.$ Observe
that their running time surprisingly decreases with increasing $p$
(for constant $p$). Understanding the best algorithms for computing
$\ell_p$-norm minimizing flows remains an exciting open direction, and
should shed light on algorithms for other network flow problems,
including maximum-flow.






\subsection{Our Results}
\paragraph{Unweighted $p$-norm flows.}
We give a new algorithm for computing unweighted $p$-norm minimizing
flows using $pm^{1+\frac{1}{p-1} + o(1)}$ operations.
%
\begin{theorem}
  \label{thm:MaxFlow}
Given a graph $G$ with $m$ edges and demands $\dd$, for any $2\leq p \leq \poly(m)$, we can find
  a flow $\ff$ satisfying the demands, i.e., $\BB^{\top} \ff = \dd$
  such that,
  \[
    \norm{\ff}_p^p \leq \left( 1 + \frac{1}{2^{\poly(\log m)}} \right) \min_{\ff:\BB^{\top}\ff = \dd}
    \norm{\ff}_p^p,
\]
in $p m^{1+ \frac{1}{p-1} + o(1)} $ arithmetic operations.
\end{theorem}
This is the fastest high-accuracy algorithm for all $p \gtrsim 5.23.$
In contrast, the previous best algorithm by~\cite{KyngPSW19} for
unweighted $p$-norm flows runs in time
$2^{O(p^{\nfrac{3}{2}})} m^{1+\frac{7}{\sqrt{p-1}} + o(1)} \poly(\log
\nfrac{1}{\eps}),$ and the one by Adil~\etal~\cite{AdilKPS19} runs in
time
$\Otil(p^{O(p)} m^{1+\frac{p-2}{3p-2}} \log^{2} \nfrac{1}{\eps}).$
Thus, for any $p \ge 2,$ we can now compute a
$(1+\frac{1}{2^{\poly(\log m)}})$-approximate unweighted $p$-norm minimizing
flows using
\begin{multline}
\min\{pm^{1+\frac{1}{p-1} + o(1)},
  \Otil(p^{O(p)} m^{1+\frac{p-2}{3p-2}}) \} \\
  \le p m^{\sqrt{5}-1+o(1)}
  \le pm^{1.24}
\end{multline}
operations. In comparison, the previous best bound was
\[\min\{2^{O(p^{\nfrac{3}{2}})} m^{1+\frac{7}{\sqrt{p-1}} + o(1)},
  \Otil(p^{O(p)} m^{1+\frac{p-2}{3p-2}}) \},
\]
which is $\Omega(m^{1.332})$ for $p \approx 444.$

\paragraph{Approximating Max-Flow.}
For $p \ge \log m,$ $\ell_p$ norms approximate $\ell_{\infty},$ and
hence the above algorithm returns an approximate maximum-flow.
For $p = \Theta\left( \frac{\log m}{\delta} \right),$ this gives an
$m^{1+o(1)}\delta^{-1}$-operations algorithm for computing a
$(1+\delta)$-approximation to maximum-flow problem on unit-capacity
graphs.
\begin{corollary}
  \label{cor:max-flow}
  Given an (undirected) graph $G$ with $m$ edges with unit capacities,
  a demand vector $\dd,$ and $\delta > 0,$ we can compute a flow $\ff$
  that satisfies the demands, i.e., $\BB^{\top} \ff = \dd$ such that
  \[\norm{\ff}_{\infty} \le (1+\delta) \min_{\ff: \BB^{\top}\ff =
      \dd} \norm{\ff}_{\infty},\]
  in $m^{1+o(1)}\delta^{-1}$ arithmetic operations.  
\end{corollary}
This gives another approach for approximating maximum flow with a
$\delta^{-1}$ dependence on the approximation achieved in the recent
works of Sherman~\cite{Sherman17} and Sidford-Tian~\cite{SidfordT18},
albeit only for unit-capacity graphs, and with a $m^{o(1)}$ factor
instead of $\poly(\log m)$. To compute max-flow essentially exactly on
unit-capacity graphs, one needs to compute $p$-norm minimizing flows
for $p=m.$


\paragraph{$p$-norm Regression and Weighted $p$-norm Flows.}
We obtain an algorithm for
solving weighted $\ell_p$-regression problems using
$pm^{\frac{p-2}{3p-2}+o(1)} \log^{2} \nfrac{1}{\eps}$ linear
solves. 
\begin{theorem}
\label{thm:lpRegression}
Given $\AA \in \mathbb{R}^{n \times m}$, $\bb \in \mathbb{R}^{n}$, for any $2 \leq p \leq \poly(m)$ and $\eps >0$, we can find $\xx \in  \mathbb{R}^{m}$ such that $\AA \xx = \bb$ and, 
\[
\norm{\xx}_p^p \leq (1+\eps) \min_{\xx:\AA\xx = \bb} \norm{\xx}_p^p,
\]
in $p m^{\frac{p-2}{3p-2}+o(1)} \log^{2} \left(\frac{1}{\eps}\right)$
calls to a linear solver.
\end{theorem}
\noindent Combined with nearly-linear time Laplacian
solvers~\cite{SpielmanT04, KoutisMP11}, we can compute
$(1+\eps)$-approximation to weighted $\ell_p$-norm minimizing flows in
$pm^{\frac{p-2}{3p-2}+1+o(1)} \log^{2} \nfrac{1}{\eps} \le
pm^{\nfrac{4}{3}+o(1)} \log^{2} \nfrac{1}{\eps}$ operations. This
gives the first high-accuracy algorithm for computing $p$ norm
minimizing flows for $p=n^{\Omega(1)}$ that runs in time $o(m^{1.5}).$

Again, for $p = \Theta(\frac{\log m}{\delta}),$ this gives an
algorithm for $\ell_{\infty}$ regression that requires
$\Otil(\delta^{-1} m^{\nfrac{4}{3}})$ calls to a linear solver,
comparable to best bound of $\Otil(\delta^{-\nfrac{2}{3}}
m^{\nfrac{4}{3}})$ by Ene and Vladu~\cite{EneV19}.

\paragraph{A caveat.}
An important caveat to the above results is that they are measuring
running time in terms of arithmetic operations and ignoring the
bit-complexity of the numbers involved. For large $p,$ even computing
the $p$-norm of a vector $\xx$ can involve numbers with large bit
complexity. To the best of our knowledge, all the algorithms for
$p$-norm minimization for large $p$, including interior point methods,
need to work with numbers with large bit complexity. In finite
precision, the algorithms would lose another $\poly(p \log m)$ factor,
so we probably need to work with floating point representations to get
better dependence. 

In Section
\ref{sec:BoxConstraint}, we present an approach to ameliorate this
concern. We show that solving a quadratic minimization problem with
$\ell_{\infty}$ constraints is sufficient for solving the smoothed
$\ell_p$-norm problem up to an $m^{1/(p-1)}$ approximation factor
(Corollary~\ref{cor:lpBox}). This results in an additional factor of
$m^{1/(p-1)}$ in the runtime, which is $m^{o(1)}$ for all
$p = \omega(1).$ Such $\ell_{\infty}$-box constrained problems have
been previously studied~\cite{CohenMTV17}.  As a result of this
reduction, we can avoid computing the $p$ powers in the objective, and
hence avoid associated numerical precision issues. Note that this
doesn't solve the bit complexity entire since we still need to compute
the gradient and the quadratic terms, which also involve large powers.

For the remaining paper, we will focus on the number of arithmetic
operations required, and leave the bit complexity question for future
work.

%

\subsection{Technical Contribution}
Our key technical tool is the notion of quadratically smoothed
$\ell_p$-norm minimization.  \cite{AdilKPS19} defined this notion and
showed that they allow us to \emph{iteratively refine} an
$\ell_p$-norm problem, i.e., given any initial solution to the
$\ell_p$-norm problem, you can make $2^{-O(p)}$ progress towards the
optimum by solving a smoothed $\ell_p$-norm problem. \cite{KyngPSW19}
combined this with graph-theoretic adaptive preconditioners to give
almost-linear time high-accuracy algorithms for computing $p$-norm
minimizing flows on unit-weighted graphs.  In~\cite{AdilPS19}, the
authors improved the iterative refinement to make $\Omega(p^{-1})$
progress rather than $2^{-O(p)}$ as in \cite{AdilKPS19, KyngPSW19}.



In this paper, we expand on the power of smoothed $p$ norm regression
problems.  Our key technical contribution is to show that smoothed $p$-norm regression problems are
inter-reducible for different $p.$ Specifically, we show that a
smoothed $p$-norm regression problem can be solved to high accuracy
using roughly $pm^{\max\{\frac{1}{p-1},\frac{1}{q}\}}$ calls to an
approximate solver for a smoothed $q$-norm problem. This is
surprising because the naive reduction from standard $p$-norm
minimization problem to standard $q$-norm minimization suffers a loss
of $m^{\frac{p}{q}-1}$ in the approximation, and no reduction
achieving high-accuracy solutions is known for standard $p$-norm
minimization problems.
%
%
\begin{theorem}
\label{thm:MainThm}
Given $\AA \in \mathbb{R}^{n \times m}$, $\bb \in \mathbb{R}^{n}$ and access
to an oracle that can solve smoothed $q$-norm problems to a constant
accuracy, for any $2 \leq p \leq \poly(m)$ and $\eps >0$, Algorithm \ref{alg:Largep} ($p > q$) and Algorithm \ref{alg:Smallp}($p<q$) find $\xx \in \mathbb{R}^{m}$ such that
$\AA \xx = \bb$ and,
\[
\norm{\xx}_p^p \leq (1+\eps) \min_{\xx:\AA\xx = \bb} \norm{\xx}_p^p,
\]
in
$\Otil \left(p m^{\max \left\{\frac{1}{q},\frac{1}{p-1}\right\}} \log^{2}
  \left(\frac{1}{\eps}\right)\right)$ calls to the smoothed $q$-oracle.
\end{theorem}

A second key idea for improving the $\poly(p)$ dependence of the
algorithm for large $p$ to a linear dependence is to use homotopy on
$p.$ This means that we successively obtain approximations for
$k$-norm minimization problems for
$k = 2^{-j}p, 2^{-j +1} p, \ldots, \nfrac{p}{2}, p.$ Note that each of
these $k$-norm minimization problems is solved via a reduction to the
smoothed $q$-norm problem for the same $q.$ Without using homotopy, the
above theorem can be proved with a quadratic dependence on $p$. The
improvement from quadratic to linear is crucial for obtaining the
$\delta^{-1} m^{1+o(1)}$ algorithm for $(1+\delta)$-approximating
maximum flow on undirected unit-capacity graphs.

The above two ideas combined allow us to solve a smoothed $p$-norm
regression problem using
$p m^{\max\{\frac{1}{q},\frac{1}{p-1}\}} \log^{2} \nfrac{m}{\eps}$
calls to the smoothed-$q$ norm solver.

Combining this reduction with the algorithm for unweighted
$\ell_p$-norm flows from~\cite{KyngPSW19}, we obtain our main result
on unweighted $\ell_p$-norm flow minimization.
Alternatively, combining the reduction with the algorithm for weighted
$\ell_p$-norm regression from~\cite{AdilKPS19}, gives us our algorithm
for weighted $\ell_p$-norm regression.

\section{Preliminaries}
In the paper, we will work in the ambient space $\rea^{m}.$ The matrix
$\AA \in \rea^{n \times m},$ with $n \le m$ will denote a constraint
matrix. Vectors will be denoted by $\Delta$ or using bold small letters such as $\xx,\bb,\dd, \ff.$
\begin{definition}[$p$-norm Problem]  For any $p$ and any
  $\AA \in \rea^{n \times m}, \bb \in \rea^{n}$, we refer to the
  following problem as a $p$-norm problem,
  \begin{equation}
    \label{eq:pNorm}
    \min_{\xx \in \rea^{m}: \AA\xx = \bb} \norm{\xx}_p^p. 
  \end{equation}
  Let $\eps >0$. A $(1+\eps)$-approximate solution to the above problem is some $\tilde{\xx}$ such that $\AA\tilde{\xx} = \bb$ and,
  \[
 \norm{\tilde{\xx}}_p^p \leq (1+\eps) \min_{\xx \in \rea^{m}: \AA\xx = \bb} \norm{\xx}_p^p.
  \]
\end{definition}

\begin{definition}[Smoothed $p$-norm Problem] For any $p \ge 2$ and
  any matrix $\AA \in \rea^{n \times m}$, vectors $\bb, \rr$ and scalar $s$, we refer to the following problem as a smoothed $p$-norm problem,
\begin{equation}
\label{eq:pNormSmoothed}
\min_{\xx: \AA\xx = \bb} \sum_e \rr_e \xx_e^2 + s \norm{\xx}_p^p. 
\end{equation}
Let $\kappa \geq 1$. A $\kappa$-approximate solution to the above problem is $\tilde{\xx}$ such that $\AA\tilde{\xx} = \bb$ and,
\[
\sum_e \rr_e \tilde{\xx}_e^2 + s \norm{\tilde{\xx}}_p^p \leq \kappa  \cdot \left(\min_{\xx: \AA\xx = \bb} \sum_e \rr_e \xx_e^2 + s \norm{\xx}_p^p\right). 
\]
\end{definition}

\begin{definition}[Residual Problem] At a given $\xx$ we define the residual problem for the $p$-norm problem \eqref{eq:pNorm} to be,
\begin{align}
\label{eq:residual}
\begin{aligned}
\max_{\Delta} \quad & \gg^{\top}\Delta -2 \sum_e \rr_e \Delta_e^2 - \norm{\Delta}_p^p\\
& \AA\Delta = 0.
\end{aligned}
\end{align}
Here $\gg = |\xx|^{p-2}\xx$ and $\rr = |\xx|^{p-2}$. We denote the objective of the residual problem at $\Delta$ by $\residual_p(\Delta)$. For $\kappa \geq 1$, a $\kappa$ approximate solution to the above residual problem is $\Dtil$ such that $\AA\Dtil = 0$ and 
\[
\residual_p(\Dtil) \geq \frac{1}{\kappa} \residual_p(\Dopt),
\]
where $\Dopt$ is the optimum of the residual problem.
\end{definition}

Note that this definition is equivalent to the definition from
\cite{AdilPS19}, which can be obtained by replacing $\Delta$ by
$p \Delta.$

\begin{definition}[Smoothed $q$-Oracle] We define the smoothed $q$-oracle to
  be an algorithm that can solve the family of smoothed $q$-norm
  problems \eqref{eq:pNormSmoothed}, to a constant approximation.
Here $\AA$ is any matrix, $\bb$ and $\rr$ are any vectors and  $s$ is any scalar.
\end{definition}


\begin{definition}[Unweighted $\ell_p$-norm flows]
Let $G$ be an unweighted graph, $\BB$ denote its edge-vertex incidence matrix and $\dd$ be a demand vector such that $\dd^{\top}\vone = 0$. 
We define the unweighted $\ell_p$-norm minimizing flow problem to be,
\[
\min_{\ff:\BB^{\top}\ff = \dd} \norm{\ff}_p^p.
\]
\end{definition}

\begin{definition}[Weighted $\ell_p$-norm flows]
Let $G$ be a weighted graph with edge weights $\ww$, $\BB$ denote its edge-vertex incidence matrix and $\dd$ be a demand vector such that $\dd^{\top}\vone = 0$. 
We define the weighted $\ell_p$-norm minimizing flow problem to be,
\[
\min_{\ff:\BB^{\top}\ff = \dd} \sum_e \ww_e|\ff_e|^p.
\]
\end{definition}
When all the edge weights in the graph are $1$, the weighted and unweighted $\ell_p$-norm flow problems are the same.

\subsubsection*{Notation}
We will use the above problem definitions for parameters $q$ and $k$ as well, where the problem is the same except we replace the $p$ with $q$ or $k$ respectively. For any definition in the following sections, we always have $q$ as a fixed variable, however $k$ and $p$ might be used as parameters interchangeably in the definitions. We always want to finally solve the $p$-norm problem. To do this we might use an intermediate parameter $k$ and solve the $k$-norm problem first. In order to solve any of these problems, we will always use a smoothed $q$-oracle and use this solution as an approximation for the  $p$-norm, $k$-norm or any other norm problem. 



\section{Solving $p$-norm Regression using Smoothed $q$-Oracles}
\label{sec:p-to-q}
We present an algorithm to solve the $p$-norm minimization problem
\eqref{eq:pNorm} using an oracle that approximately solves a
smoothed $q$-norm problem \eqref{eq:pNormSmoothed}. We use the main
iterative refinement procedure from \cite{AdilPS19} as the base
algorithm, and show that approximately solving  a smoothed $q$-norm
problem suffices to obtain an approximation for the residual problem for $p$-norms. The following results formalize this.


\begin{restatable}{lemma}{ApproxLargep}
\label{lem:ApproxLargep}
Let $k\geq q$ and $\nu$ be such that
$\residual_k(\Dopt) \in (\nu/2,\nu]$, where $\Dopt$ is the optimum of
the residual problem for $k$-norm \eqref{eq:residual}. The following problem has optimum
at most $\nu$.
\begin{align}
\label{eq:qProblemLargep}
\begin{aligned}
  \min_{\Delta} \quad & \sum_e \rr_e \Delta_e^2 + \frac{1}{2} \left(
    \frac{\nu}{m} \right)^{1-\frac{q}{k}} \norm{\Delta}_q^q\\
  &\gg^{\top} \Delta =\nu/2\\
  &\AA \Delta = 0.
\end{aligned}
\end{align}
For $\beta \geq 1$, if $\Dtil$ is a
feasible solution to the above problem such that the objective is at
most $\beta \nu$, then the following holds,
\[
  2\sum_e \rr_e (\alpha \Dtil)_e^2 + \norm{\alpha\Dtil}_{k}^{k} \leq
  \alpha \frac{\nu}{4},
\]
where $\alpha =
\frac{1}{16\beta}m^{-\frac{k}{k-1}\left(\frac{1}{q}
    - \frac{1}{k}\right)}$. 
\end{restatable}

\begin{restatable}{corollary}{ApproxToRes}
\label{cor:ApproxToRes}
Let $k\geq q$ and $\nu$ be such that
$\residual_k(\Dopt) \in (\nu/2,\nu]$, where $\Dopt$ is the optimum of
the residual problem \eqref{eq:residual} for $k$-norm. For $\beta \geq 1$, if $\Dtil$ is a
feasible solution to \eqref{eq:qProblemLargep} such that the objective of \eqref{eq:qProblemLargep} at $\Dtil$ is at
most $\beta \nu$, then $\alpha \Dtil$ gives an $O\left(\beta m^{\frac{k}{k-1}\left(\frac{1}{q}
    - \frac{1}{k}\right)}\right)$-approximate solution to the residual problem \eqref{eq:residual}, where $\alpha = \frac{1}{16\beta}m^{-\frac{k}{k-1}\left(\frac{1}{q}
    - \frac{1}{k}\right)}$.
\end{restatable}

\begin{restatable}{lemma}{ApproxSmallp}
\label{lem:ApproxSmallp}
Let $2 \le k \le q$ and $\nu$ be such that
$\residual_k(\Dopt) \in (\nu/2,\nu]$ where $\Dopt$ is the optimum of
the residual problem for $k$-norm \eqref{eq:residual}. The following problem has optimum at most $\nu$. 
\begin{align}
\label{eq:qProblemSmallp}
\begin{aligned}
\min_{\Delta} \quad &  \sum_e \rr_e \Delta_e^2 + \frac{\nu^{1 - q/k}}{2^{q/k}}\norm{\Delta}_q^q  \\
 & \gg^{\top} \Delta = \nu/2\\
& \AA \Delta = 0.
\end{aligned}
\end{align}
For $\beta \geq 1$, if $\Dtil$ is a feasible solution to the above problem such that the objective is at
most $\beta \nu$, then $\alpha \Dtil$, where $\alpha = \frac{1}{16\beta}m^{-\frac{k}{k-1}\left(\frac{1}{k}
    - \frac{1}{q}\right)}$, gives an
$O\left(\beta m^{\frac{k}{k-1}\left(\frac{1}{k} -
      \frac{1}{q}\right)}\right)$-approximate solution to the residual
problem \eqref{eq:residual} for $k$-norm.
\end{restatable}

We will now prove Theorem
\ref{thm:MainThm} by considering the two cases, $p>q$ and $p<q$
separately. Let us first look at the case where $p>q$.

\begin{algorithm}
\caption{Solving the residual problem using smoothed $q$-oracle}
\label{alg:Subroutine}
 \begin{algorithmic}[1]
 \Procedure{Residual}{$\xx^{(0)},\AA, \bb, k, \beta$}
 \State $\xx \leftarrow \xx^{(0)}$
  \State $ T \leftarrow O\left(ckm^{\frac{k}{k-1}\left(\frac{1}{q} - \frac{1}{k}\right)}\log \frac{m}{\beta-1}\right)$
 
 \For{$t = 1: T$}
 \For{$ i \in \left[\log \left(\frac{(\beta - 1)\|\xx^{(0)}\|_{k}^{k}}{k
         m}\right),\log \left(\|\xx^{(0)}\|_{k}^{k}\right)\right]$}\label{alg:line:binary}
   \State $\Dtil^{(i)} \leftarrow $ $c$-approximate
   solution of \eqref{eq:qProblemLargep} with $\nu = 2^i$, using a smoothed $q$-oracle.
   \State $\alpha \leftarrow \frac{1}{16c}m^{-\frac{k}{k-1}\left(\frac{1}{q} - \frac{1}{k}\right)}$
   \EndFor 
   \State $i \leftarrow \arg\min_{i} \norm{\xx-\alpha\frac{\Dtil^{(i)}}{k}}_k^k$ 
   \State $\xx \leftarrow \xx -\alpha\frac{\Dtil^{(i)}}{k}$ 
   \EndFor 
    \State \Return $\xx$ 
   \EndProcedure
 \end{algorithmic}
\end{algorithm}

\begin{algorithm}
\caption{Algorithm for $p>q$}
\label{alg:Largep}
 \begin{algorithmic}[1]
 \Procedure{pNorm}{$\AA, \bb, \epsilon$}
 \State $\xx \leftarrow$ $O(1)$-approximation to $\min_{\AA\xx = \bb} \norm{\xx}^q_q$
 \State $k \leftarrow 2q$
 \While{$k\leq p$}
  \State $ \xx \leftarrow$ {\sc Residual}$(\xx,\AA,\bb,k, O(1))$
 
   \State $k \leftarrow 2k$ 
    \EndWhile 
    \State $ \xx \leftarrow$ {\sc Residual}$(\xx,\AA,\bb,p,1+\eps)$ 
   \State \Return $\xx$ 
   \EndProcedure
 \end{algorithmic}
\end{algorithm}

\subsection{$p>q$}
We use a homotopy approach to solve such problems, i.e., we start with
a solution to the $q$-norm problem \eqref{eq:pNorm}, and successively
solve for $2q$-norms, $2^2q$-norms,$...$,$p$-norms, using the previous
solution as a starting solution. This can be done without homotopy and
directly for $p$-norms however, with homotopy we achieve the
dependence on $p$ to be linear which otherwise would have been $p^2$. To this end, we will first show that
for any $p>q$, given a constant approximate solution to the $p$-norm
problem we can find a constant approximate solution to the $2p$-norm
problem in $O\left(pm^{\frac{2p}{2p-1}\left(\frac{1}{q} -
      \frac{1}{2p}\right)}\log m \log (pm)\right)$ calls to the
smoothed $q$-oracle.

\begin{restatable}{lemma}{ptotwop}
\label{lem:p-to-2p}
Let $p \geq q$. Starting from $\xx^{(0)}$, an $O(1)$-approximate solution to the $p$-norm problem \eqref{eq:pNorm}, Algorithm \ref{alg:Subroutine} finds an $O(1)$-approximate solution to the $2p$-norm problem \eqref{eq:pNorm} in 
\[
O\left(pm^{\frac{2p}{2p-1}\left(\frac{1}{q} - \frac{1}{2p}\right)}\log m \log (pm)\right)
\] calls to a smoothed $q$-oracle. 
\end{restatable}
In order to prove Lemma \ref{lem:p-to-2p}, we need the following lemmas. The first is an application of the iterative
refinement scheme from \cite{AdilPS19} \footnote{This version spells out some more details. Note the additional factor $p$ in the iterative step. This comes from the fact that we have scaled the residual problem to absorb the $p$-factors.}.

\begin{lemma}[Iterative Refinement \cite{AdilPS19}]
\label{lem:OriginalIterativeRefinement}
Let $\eps >0,  p \geq 2$, and $\kappa \geq 1$. Starting from $\xx^{(0)}$, and iterating as, $\xx^{(t+1)} = \xx^{(t)} -\Delta/p$, where 
$\Delta$ is a $\kappa$-approximate solution to the residual problem \eqref{eq:residual}, we get an $(1+\eps)$-approximate solution to \eqref{eq:pNorm} in at most
$O\left(p \kappa \log \left(\frac{\|\xx^{(0)}\|^p_p - OPT}{\eps OPT}\right)\right)$ calls to a $\kappa$-approximate solver for the residual problem.
\end{lemma}

We have deferred the proofs of these lemmas to the appendix.

\begin{restatable}{corollary}{IterativeRefinement}
\label{lem:IterativeRefinement}
Let $p \geq 2$ and $\kappa \geq 1$. Starting from $\xx^{(0)}$, an $O(1)$-approximate solution to the $p$-norm problem \eqref{eq:pNorm}, and iterating as $\xx^{(t+1)} = \xx^{(t)}-\Delta/p$, where $\Delta$ is a $\kappa$-approximate solution to the residual problem for the $2p$-norm \eqref{eq:residual}, we get an $O(1)$-approximate solution for the $2p$-norm problem in at most $O\left(\kappa p \log m\right)$ calls to a $\kappa$-approximate solver for the residual problem.
\end{restatable}

The next lemma bounds the range of binary search in the algorithm.
\begin{restatable}{lemma}{BinarySearch}
\label{lem:BinarySearch}
Let $k \leq r$ and $\xx^{(0)}$ be an $O(1)$-approximate solution to the $k$-norm
problem \eqref{eq:pNorm} and assume that $\xx^{(0)}$ is not an $\alpha$-approximate
solution for the $r$-norm problem. For some
\[
  \nu \in \left[\Omega(1)(\alpha-1)\frac{\|\xx^{(0)}\|_{r}^{r}}{rm^{\left(\frac{r}{k}-1\right)}},\|\xx^{(0)}\|_{r}^{r}\right],
\]
$\residual_r(\Dopt) \in (\nu/2,\nu]$, where $\Dopt$ is the optimum of
the residual problem for the $r$-norm problem \eqref{eq:residual}.
\end{restatable}

Using Lemma \ref{lem:BinarySearch}, and Corollaries \ref{lem:IterativeRefinement} and  \ref{cor:ApproxToRes}, we can now prove Lemma \ref{lem:p-to-2p}.

\subsubsection*{Proof of Lemma \ref{lem:p-to-2p}}
\begin{proof}
  From Corollary \ref{lem:IterativeRefinement} we know that we need to
  solve the residual problem to a $\kappa$ approximation
  $O\left(\kappa p \log m\right)$ times. Corollary \ref{cor:ApproxToRes}
  shows that for some $\nu$ solving problem \eqref{eq:qProblemLargep}
 up to a constant approximation gives an
  $O\left(m^{\frac{2p}{2p-1}\left(\frac{1}{q} -
        \frac{1}{2p}\right)}\right)$-approximate solution to the
  residual problem for the $2p$-norm problem. Note that we have only $O\left(\log(p m)\right)$ values for
  $\nu$, (Lemma \ref{lem:BinarySearch}). The total number of calls to the smoothed $q$-oracle are
  therefore
  \[O\left(p m^{\frac{2p}{2p-1}\left(\frac{1}{q} -
        \frac{1}{2p}\right)}\log(m)\log(pm)\right).\]
\end{proof}

We can now prove our main result Theorem \ref{thm:MainThm}.
\subsubsection*{Proof of Theorem \ref{thm:MainThm} for $p>q$}
\begin{proof}
  We start with an $O(1)$ approximate solution to the $q$ norm problem 
  and successively solve for
  $2q,2^2q,...,r$, where $r$ is such that $ p/2< r \leq p$ . From Lemma
  \ref{lem:p-to-2p}, the total number of iterations required to get an
  $O(1)$-approximate solution to the $r$ norm problem where $k = 2^i
  q$ is,
  \begin{align*}
   & \sum_{\substack{0 \le i \le \log_2 r \\ k = q2^{i}}}
    O\left(km^{\frac{k}{k-1}\left(\frac{1}{q} -
    \frac{1}{k}\right)}\log m \log (km)\right)\\
    & \leq O\left(\log m \log (pm) \right)
          \sum_{\substack{0 \le i \le \log_2 r \\ k = q2^{i}}}
      km^{\frac{k}{k-1}\left(\frac{1}{q} - \frac{1}{k}\right)} \\
    & \leq O\left(pm^{\frac{1}{q}}\log m \log (pm) \right).
  \end{align*}
  We now have a constant approximate solution to the $r$ norm
  problem. We need to find a $(1+\eps)$-approximate solution to the
  $p$-norm problem. To do this, we use the  iterative refinement
  scheme for $p$-norms from \cite{AdilPS19}, to obtain a $p$-norm residual problem at every iteration. We solve the $p$-norm residual problem, by doing a binary search followed by solving the corresponding smoothed $q$-oracle. From Lemma \ref{lem:BinarySearch} we know that we only have to search over at most $O\left(\log \frac{pm}{\eps}\right)$ values of $\nu$. From Corollary \ref{cor:ApproxToRes}, we obtain an $O\left(m^{\frac{p}{p-1}\left(\frac{1}{q} - \frac{1}{p}\right)}\right) \leq O\left(m^{\frac{1}{q}}\right)$-approximate solution to the residual problem for $p$-norms. We thus have an additional, $O\left(pm^{\frac{1}{q}}\log^2 \frac{pm}{\eps}\right)$ iterations from the iterative refinement, giving us a total of $\Otil\left(pm^{\frac{1}{q}}\log^2 \frac{1}{\eps}\right) $ iterations.
\end{proof}


%
%

\begin{algorithm}
\caption{Algorithm for $p<q$}
\label{alg:Smallp}

 \begin{algorithmic}[1]
 \Procedure{pNorm}{$\AA, \bb, \epsilon$}
 \State $\xx^{(0)} \leftarrow$ $O(1)$-approximation to $\min_{\AA\xx = \bb} \norm{\xx}^q_q$
 \State $\xx \leftarrow \xx^{(0)}$
  \State $ T \leftarrow O\left(cpm^{\frac{1}{p-1}}\log \frac{m}{\eps}\right)$
 \For{$t = 1: T$}
 \For{$ i \in \left[\log \left(\frac{\eps\|\xx^{(0)}\|_{p}^{p}}{p
         m}\right),\log \left(\|\xx^{(0)}\|_{p}^{p}\right)\right]$}
   \State $\Dtil^{(i)} \leftarrow $ $c$-approximate
   solution of \eqref{eq:qProblemSmallp} with $\nu = 2^i$, using a smoothed $q$-oracle.
   \State $\alpha \leftarrow \frac{1}{16c} m^{-\frac{p}{p-1}\left(\frac{1}{p} - \frac{1}{q}\right)}$
   \EndFor 
   \State $i \leftarrow \arg\min_{i} \norm{\xx-\alpha\frac{\Dtil^{(i)}}{p}}_p^p$ 
   \State $\xx \leftarrow \xx -\alpha\frac{\Dtil^{(i)}}{p}$ 
   \EndFor 
    \State \Return $\xx$ 
 \EndProcedure 
 \end{algorithmic}
\end{algorithm}

\subsection{$p<q$}
We will now prove Theorem \ref{thm:MainThm} for $p<q$. For this case, we do not require any homotopy approach. We can directly solve the $p$-norm problem using the smoothed $q$-oracle. We will first prove the following lemma.
\begin{lemma}
  \label{lem:p-to-q}
  Let $2 \leq p < q$. Starting from $\xx^{(0)}$, an $O(1)$-approximate
  solution to the $q$-norm problem \eqref{eq:pNorm}, we can find an $O(1)$-approximate
  solution to the $p$-norm \eqref{eq:pNorm} problem in
  $O\left(pm^{\frac{1}{p-1}}\log m \log pm\right)$ iterations. Each
  iteration solves a 
  smoothed $q$-norm problem \eqref{eq:pNormSmoothed}.
\end{lemma}

We need the following lemmas to prove Lemma \ref{lem:p-to-q}. The proofs are deferred to the appendix. We begin by using the following version of iterative refinement.

\begin{restatable}{lemma}{IterativeRefinementSmall}
\label{lem:IterativeRefinementSmall}
Let $p \geq 2$ and $\kappa \geq 1$. Starting from $\xx^{(0)}$, an $O(1)$-approximate solution to the $q$-norm problem \eqref{eq:pNorm}, and iterating as $\xx^{(t+1)} = \xx^{(t)}-\Delta/p$, where $\Delta$ is a $\kappa$-approximate solution to the residual problem for $p$-norm \eqref{eq:residual}, we get an $O(1)$-approximate solution for the $p$-norm  problem \eqref{eq:pNorm} in at most $O\left(\kappa p \log m\right)$ calls to a $\kappa$-approximate solver for the residual problem.
\end{restatable}

\begin{restatable}{lemma}{BinarySearchSmall}
\label{lem:BinarySearchSmall}
Let $p<q$ and $\xx^{(0)}$ be an $O(1)$-approximate solution to the $q$-norm problem \eqref{eq:pNorm}. Assume that $\xx^{(0)}$ is not an $\alpha$-approximate solution for the $p$-norm problem \eqref{eq:pNorm}. For some 
\[\nu \in \left[ \Omega(1)(\alpha-1)\frac{\|\xx^{(0)}\|_{p}^{p}}{p m},\|\xx^{(0)}\|_{p}^{p}\right],\]
 $\residual_p(\Dopt) \in (\nu/2,\nu]$, where $\Dopt$ is the optimum of the residual problem for the $p$-norm problem \eqref{eq:residual}.
\end{restatable}

We can now prove Lemma \ref{lem:p-to-q}.

\subsubsection*{Proof of Lemma \ref{lem:p-to-q}}
\begin{proof}
From Lemma \ref{lem:IterativeRefinementSmall} we know that we need to solve the residual problem \eqref{eq:residual} to a $\kappa$ approximation $O\left(\kappa p \log m\right)$ times. Lemma \ref{lem:ApproxSmallp} shows that for some $\nu$ solving problem \eqref{eq:qProblemSmallp} gives an $O\left(m^{\frac{p}{p-1}\left(\frac{1}{p} - \frac{1}{q}\right)}\right)$-approximate solution to the residual problem. From Lemma \ref{lem:BinarySearch} we have only $O\left(\log(p m)\right)$ values for $\nu$, giving a total iteration count as,
 \begin{equation}
 O\left(p m^{\frac{p}{p-1}\left(\frac{1}{p} - \frac{1}{q}\right)}\log(m)\log(pm)\right)
  \leq O\left(p m^{\frac{1}{p-1}}\log(m)\log(pm)\right).\end{equation}
\end{proof}

Lemma \ref{lem:p-to-q} implies the remaining part of Theorem \ref{thm:MainThm}.
\subsubsection*{Proof of Theorem \ref{thm:MainThm} for $p<q$.}
\begin{proof}
We start with a constant approximate solution to the $q$-norm problem. Starting from this solution we can use the iterative refinement procedure on $p$-norms from \cite{AdilPS19} to get a $p$-norm residual problem \eqref{eq:residual} at every iteration. Now, in order to solve this residual problem, we do a binary search over its values $\nu$, which are only $O\left(\log \frac{pm}{\eps}\right)$ values. Now for each value $\nu$, we can solve a $q$-norm smoothed problem \eqref{eq:qProblemSmallp} to get an $O\left(m^{\frac{p}{p-1}\left(\frac{1}{p} - \frac{1}{q}\right)}\right) \leq O\left(m^{\frac{1}{p-1}}\right)$-approximate solution to the $p$-norm residual problem (Lemma \ref{lem:ApproxSmallp}). Therefore, we have a total iteration count of $\Otil\left(pm^{\frac{1}{p-1}}\log^2 \frac{1}{\eps}\right)$.
%
\end{proof}

%
%
%
%
%
%



\section{Algorithm for Unweighted $p$-Norm-Flow}
\label{sec:unweighted-flow}
In this section we will prove Theorem \ref{thm:MaxFlow} and Corollary~\ref{cor:max-flow}. Our main
algorithm will be Algorithm \ref{alg:Largep}, and we will use the
algorithm from \cite{KyngPSW19} for $p$-norm minimization as our
smoothed $q$-oracle, for $q = \sqrt{\log m}.$ The following theorem
from~\cite{KyngPSW19} gives the guarantees of the algorithm, though
the running time is spelled out in more detail, and it is stated for a
slightly improved error bound from $\frac{1}{\poly(m)}$ to
$2^{-\poly(\log m)}$ since that does not increase the running time of
the algorithm significantly.

\begin{theorem}[Theorem 1.1,\cite{KyngPSW19}]
  \label{thm:qOracle}
  We're given $p \geq 2$, weights $\rr \in \rea^E_{\geq 0}$, a
  gradient $\gg \in \rea^E,$ a demand vector $\bb \in \rea^V$ with
  $\bb^{\top}\vone = 0,$ a scalar $s,$ and an initial solution
  $\ff^{(0)}.$ Let
  $val(\ff) = \gg^{\top}\ff + \sum_e \rr_e \ff_e^{2} +
  s \norm{\ff}_{p}^{p}$ and let $\opt$ denote
  $\min_{\ff: \BB^{\top}\ff = \bb} val(\ff).$

  If all parameters are bounded between
  $[2^{-\poly(\log m)}, 2^{\poly(\log m)}]$, we can compute a flow
  $\tilde{\ff}$ satisfying demands $\bb$, i.e., $\BB^{\top}\ff = \bb$
  such that
  \begin{equation}
    val(\ff) - \opt \leq \frac{1}{2^{\poly(\log m)}} (val(\ff^{(0)}) -
    \opt) + \frac{1}{2^{\poly(\log m)}},
  \end{equation}
  in $2^{O(p^{\nfrac{3}{2}})} m^{1+ \frac{7}{\sqrt{p-1}} + o(1)}$ time
  where $m$ is the number of edges in $G$.
\end{theorem}

We now give the proof of Theorem~\ref{thm:MaxFlow}. We will assume
that the optimum of the initial $p$-norm flow problem is at most
$O(m)$ and at least a constant. We next show why this is a valid
assumption. For $p \geq q$, we would have to use a homotopy approach,
i.e., start with an $O(1)$-approximate solution to the $q$-norm
problem and proceed by solving the $k$-norm problem to an
$O(1)$-approximation for $k = 2q,2^2q,...p$. For every $k$, the
initial solution is at most a factor $m$ away from the optimum (To see
this, refer to the proof of Lemma \ref{lem:BinarySearch}). Therefore,
at every $k$, we can scale the problem so that the objective evaluated
at the initial solution is $\Theta(m),$ and the optimum is guaranteed
to be at least a constant. When $p \leq q$, a constant approximate solution to the $q$-norm problem directly gives an $O(m)$-approximate solution to the $p$-norm problem, and we can similarly scale it.

%
\begin{proof}
  We will use Theorem~\ref{thm:MainThm} to reduce solving $p$-norm
  problems to obtain constant approximate solutions to smoothed
  $q$-norm problems for $q = \max\{2, \sqrt{\log m}\}.$ These smoothed
  $q$-norm problems are of the form Problem~\eqref{eq:qProblemLargep}, for some $k = 2^i q$ when $p\geq q$ (note that we are using homotopy here), or Problem~\eqref{eq:qProblemSmallp} with $k$ replaced by $p$ when $p \leq q$. We
  will use the algorithm from~\cite{KyngPSW19} as the oracle to solve these
  problems to constant accuracy.

  Observe that this oracle requires $m^{1+o(1)}$ time for
  approximately solving smoothed $q$-norm problems. When $p \geq q$,
  Theorem \ref{thm:MainThm} implies that we require
  $p m^{\frac{1}{\sqrt{\log m}}} \poly(\log m)$ calls to the oracle to
  solve the problem to a $2^{-\poly(\log m)}$ accuracy, giving us a
  total of $p m^{1+o(1)}$ operations. When $p <q$, again from
  Theorem~\ref{thm:MainThm}, we require,
  $p m^{\frac{1}{p-1}} \poly(\log m)$ calls to the oracle giving us a
  total $p m^{1+ \frac{1}{p-1} + o(1)}$ operations.

  Thus, it suffices to show that we can use the algorithm
  from~\cite{KyngPSW19} to solve the smoothed $q$-norm
  problems. Ideally, we would have liked to convert
  Problems~\eqref{eq:qProblemLargep} and~\eqref{eq:qProblemSmallp}
  directly into problems of the form that can be solved using
  Theorem~\ref{thm:qOracle}. However, due to some technical
  difficulties, we will bypass this and directly show that we can
  obtain an approximate solution to the residual $k$-norm (or $p$-norm, for notational convenience we will use the parameter $k$ instead of $p$), by solving a problem of the form required by Theorem~\ref{thm:qOracle}. For $p \geq q$, we have the following result. 
  \begin{restatable}{lemma}{qSTOCoracle}
\label{lem:appendix:qSTOCoracle}
Let $p\geq q \geq 2$ and $\nu$ be such that
$\residual_p(\Dopt) \in (\nu/2,\nu]$, where $\Dopt$ is the optimum of
the residual problem for $q$-norm. The following problem has optimum
at most $-\frac{\nu}{4}.$
\begin{align}
  \label{eq:qSTOC}
\begin{aligned}
  \min_{\Delta: \AA\Delta = \vzero} \gg^{\top} \Delta + 2 \sum_{e}
  \rr_e \Delta_e^{2} + \frac{1}{4} \left( \frac{\nu}{m}
  \right)^{1-\frac{q}{p}} \norm{\Delta}_{q}^{q}.    
\end{aligned}
\end{align}
If $\Dtil$ is a feasible solution to the above program such that the
objective is at most $-\frac{\nu}{16}$, then a scaling of $\Dtil$
gives us a feasible solution to $\residual_p$ with objective value
$\Omega(\nu m^{-\frac{p}{p-1}\left(\frac{1}{q}- \frac{1}{p} \right)}).$
\end{restatable}

   For $p<q$, a lemma similar to Lemma~\ref{lem:appendix:qSTOCoracle} can be shown (refer to the appendix Lemma~\ref{lem:appendix:qSTOCoracleSmallp} )and the remaining proof is similar to the following. 
  Assume that the $k$-residual problem that has been reduced from, has an
  objective value in $(\nu/2, \nu].$ Lemma~\ref{lem:appendix:qSTOCoracle} shows that solving a smoothed
  $q$-norm problem of the form given by   
  \[
  \min_{\AA\Delta = 0} -\gg^{\top}\Delta + \sum_e \rr_e
      \ff_e^{2} + s \norm{\ff}_{q}^{q}.
  \]
suffices. Note that we are using the same $\gg, \rr,$
  and $s = \frac{1}{2} \left( \frac{\nu}{m} \right)^{1-\frac{q}{k}}.$
as Problem~\eqref{eq:qProblemLargep}.  We will use $\bb = \vzero.$

Let us first see whether our above parameters are bounded between 
$[2^{-\poly(\log m)}, 2^{\poly(\log m)}]$. Note that, we have scaled the initial $k$-norm problem so that the optimum is at most $O(m)$ and at least $O(1)$. Also, we are always starting from an initial solution $\xx^{(0)}$ that gives an objective value at most $O(m)$. Now, in the first step of the iterative refinement, our parameters are $\gg = |\xx^{(0)}|^{k-2}\xx^{(0)}$, and $\rr =  |\xx^{(0)}|^{k-2}$ and we know that $\|\xx^{(0)}\|_k^k \leq O(m)$.
\begin{equation}
\|\xx^{(0)}\|_{k-2}^{k-2} \leq m^{(k-2)\left(\frac{1}{k-2} - \frac{1}{k}\right)}\|\xx^{(0)}\|_{k}^{k-2} 
\leq O(1)  m^{2/k}m^{(k-2)/k} \leq O(m),
\end{equation}
and,
\begin{equation}
 |\xx^{(0)}|^{k-2}\xx^{(0)} \leq \|\xx^{(0)}\|_{k-1}^{k-1}\\ \leq m^{(k-1)\left(\frac{1}{k-1} - \frac{1}{k}\right)}\|\xx^{(0)}\|_{k}^{k-1} \leq O(1)  m^{1/k}m^{(k-1)/k}  \leq O(m).
\end{equation}
 At every iteration $t$ of the iterative refinement, $ \gg = |\xx^{(t)}|^{k-2}\xx^{(t)}$ and $\rr =  |\xx^{(t)}|^{k-2}$, and since we guarantee that $\|\xx\|_k^k$ only decreases with every iteration, if the parameters are bounded initially, they are bounded throughout. From the above calculations, we see that $\gg$ and $\rr$ are bounded as required. Now, we are required to bound $s$. Note that since the initial objective is at most $O(m)$, the residual problem has an optimum at most $O(m)$ and therefore $\nu \leq O(m)$. So we have $s$ bounded as well.
 
 We will next show, how to get an approximation to the residual problem. We are now solving the following problem,
  \[
  \min_{\AA\Delta = 0} -\gg^{\top}\Delta + \sum_e \rr_e
      \ff_e^{2} + s \norm{\ff}_{q}^{q}.
  \]
  From Lemma~\ref{lem:appendix:qSTOCoracle} we know that the optimum of the above problem is at most $-\nu/4$.
  We can now use the guarantees
  from~\ref{thm:qOracle} for the algorithm from~\cite{KyngPSW19},
  starting from the flow $\ff^{(0)} = \vzero,$ to find a flow
  $\tilde{\ff}$ such that,
  \begin{align*}
    val(\tilde{\ff}) 
    & \leq \frac{1}{2^{\poly(\log m)}}val(\ff^{(0)}) +  \left( 1 -\frac{1}{2^{\poly(\log m)}} \right) \opt
       +\frac{1}{2^{\poly(\log m)}} \\
    &\leq  0 + \frac{-\nu}{4} \cdot \frac{1}{2}
      +\frac{1}{2^{\poly(\log m)}} \le  \frac{-\nu}{16}. 
  \end{align*}
We got the last inequality by using, $\frac{1}{2^{\poly(\log m)}} \leq \nu/16$. Note that $\nu \geq \eps OPT/km \geq \eps /pm \geq 1/2^{\poly(\log m)}$, where $OPT$ is the optimum of the $k$-norm problem. We can now use Lemma~\ref{lem:appendix:qSTOCoracle} to get a $m^{-\frac{k}{p-1}\left(\frac{1}{q}- \frac{1}{k} \right)}$-approximate solution to the $k$-residual problem as required.
 \end{proof}
 
\subsubsection*{Proof of Corollary \ref{cor:max-flow}}
\begin{proof}
  We will show that, for $p \ge \log m,$ $\ell_p$ norms approximate
  $\ell_{\infty}.$ The algorithm from Theorem~\ref{thm:MaxFlow}
  returns a flow $\ff$ in $pm^{1+o(1)}$ operations such that
\[\norm{\ff}_\infty \le \norm{\ff}_{p} \le 2^{\nfrac{1}{p}}
  \min_{\ff:\BB^{\top} \ff = \dd} \norm{\ff}_p \le (2m)^{\nfrac{1}{p}}
  \min_{\ff:\BB^{\top} \ff = \dd} \norm{\ff}_{\infty}. \] Thus, for
$p = \Theta\left( \frac{\log m}{\delta} \right),$ this is a
$m^{1+o(1)}\delta^{-1}$-operations algorithm for computing a
$(1+\delta)$-approximation to maximum-flow problem on unweighted
graphs.
\end{proof}

\section{$\ell_p$-Regression}
\label{sec:regression}
In this section, we will prove Theorem \ref{thm:lpRegression}. Our
base algorithm would be the iterative refinement scheme, followed by
solving the smoothed $q$-norm problem to approximate the residual
problem. To solve the smoothed $q$-norm problem, we will use the
algorithm {\sc Gamma-Solver} from \cite{AdilKPS19}. This algorithm has
a runtime dependence of $q^{O(q)}$, however our choice of $q =
\max\{2,\sqrt {\log m}\}$ means that these factors are $m^{o(1)}.$
To begin with, we want to solve the general $\ell_p$-norm regression problem to 
$(1+\eps)$ approximation,
\[
\min_{\xx:\AA\xx=\bb} \norm{\xx}_p^p.
\]

Consider problem \eqref{eq:qProblemLargep} and a change of variable, $\zeta = \nu^{\frac{1}{q}-\frac{1}{p}}m^{-\left(\frac{1}{q}-\frac{1}{p}\right)} \Delta$. For the same parameters $\gg$, $\rr$ and $\AA$, the problem can be equivalently phrased as,
 \begin{align}
 \label{eq:qProblemScaled}
 \begin{aligned}
\min_{\zeta \in \mathbb{R}^m} \quad& \nu^{-2\left(\frac{1}{q}-\frac{1}{p}\right)}m^{2\left(\frac{1}{q}-\frac{1}{p}\right)}  \sum_e \rr_e \zeta_e^2 + \frac{1}{2}\norm{\zeta}_{q}^{q}\\
& \gg^{\top} \zeta =  m^{-\left(\frac{1}{q}-\frac{1}{p}\right)} \nu^{1+\frac{1}{q}-\frac{1}{p}}/2\\
& \AA\zeta = 0.
\end{aligned}
 \end{align}
 
 Now, we define the following function from \cite{AdilKPS19}, since we would want to use the algorithm from the paper as an oracle.
\begin{align}
  \label{eq:gamma}
  \begin{aligned}
    \gamma_q(t,x) = \begin{cases}
      \frac{q}{2}t^{q-2}x^2 & \text{if } |x| \leq t\\
      |x|^q +\left(\frac{q}{2}-1\right)t^q & \text{otherwise.}
    \end{cases}
  \end{aligned}
\end{align}
The following lemma relates the objectives of the $\gamma_q$ function and the problem \eqref{eq:qProblemScaled}.
 \begin{restatable}{lemma}{SmoothedGamma}
 \label{lem:Smoothed2Gamma}
   If the optimum of \eqref{eq:qProblemLargep} is at most $\nu$, then
   the following problem has optimum at most $2q \nu$.
   \begin{align}
     \label{eq:gammaObj}
     \begin{aligned}
       \min_{\Delta} \quad & \gamma_q(\tt,\Delta)\\
       & \gg^{\top}\Delta =  m^{-\left(\frac{1}{q}-\frac{1}{p}\right)} \nu^{1+\frac{1}{q}-\frac{1}{p}}/2\\
       & \AA\Delta = 0.
     \end{aligned}
   \end{align}
   Here $\tt = \left(\nu^{-2\left(\frac{1}{q}-\frac{1}{p}\right)}m^{2\left(\frac{1}{q}-\frac{1}{p}\right)}\right)^{1/(q-2)}|\xx|$.
   Let $s(\zeta)$ denote the objective of problem \eqref{eq:qProblemScaled} evaluated at $\zeta$. The following relation holds for any $\zeta$,
   \[
   \frac{1}{2q}\gamma_q(\tt,\zeta) \leq s(\zeta) \leq 3\gamma_q(\tt,\zeta).
   \]
 \end{restatable}
 
 The algorithm {\sc Gamma-Solver} in \cite{AdilKPS19} that minimizes the $\gamma_q$ objective, requires that the optimum is at most $1$ and the $\tt$'s are bounded as $m^{-1/q} \leq \tt \leq 1$. The next lemma shows us how to scale down the objective so as to achieve these guarantees.

\begin{restatable}{lemma}{scaling}
\label{lem:Scaling}
Let $\nu$ be such that the optimum of \eqref{eq:gammaObj} is at most $2q \nu$ and $\tt$ be as defined in Lemma \ref{lem:Smoothed2Gamma}. Let,
\[
\hat{\tt_j} = \begin{cases}
m^{-1/q} & (2q \nu)^{-1/q}\tt_j \leq m^{-1/q},\\
1 & (2q \nu)^{-1/q}\tt_j \geq 1,\\
(2q \nu)^{-1/q}\tt_j & \text{otherwise.}
\end{cases}
\]
Note that $m^{-1/q} \leq \hat{\tt} \leq 1$. The following program has optimum at most $1$.
  \begin{align}
     \label{eq:GammaScaled}
     \begin{aligned}
       \min_{\Delta}\quad & \gamma_q(\hat{\tt},\Delta) \\
       & \gg^{\top}\Delta = 2^{\frac{1}{2}-\frac{1}{q}} q^{-\frac{1}{2}-\frac{1}{q}}m^{-\left(\frac{1}{q}-\frac{1}{p}\right)} \nu^{1-\frac{1}{p}}\\
       & \AA\Delta = 0.
     \end{aligned}
   \end{align}
  If $\Dtil$ is a $\kappa$-approximate solution to \eqref{eq:GammaScaled}, then a scaling of $\Dtil$ satisfies the constraints of \eqref{eq:gammaObj} and gives the following bound on its objective,
  \[\gamma_q\left(\tt,\left(\frac{q}{2}\right)^{1/2}(2q \nu)^{1/q} \Dtil\right) \leq q^{1+q/2} \nu \kappa. \]
\end{restatable}

Algorithm {\sc Gamma-Solver} can be applied to solve the problem obtained in the previous lemma. The
following theorem \footnote{Theorem 5.3 in the proceedings version,
  Theorem 5.8 in the arxiv version. This version spells out some more
  details.} from \cite{AdilKPS19} gives the guarantees for the
algorithm.
 \begin{theorem}[\cite{AdilKPS19}]
   Let $p \geq 2$. Given a matrix $\AA$ and vectors $\xx$ and $\tt$
   such that $\forall e, m^{-1/p} \leq \tt_e \leq 1$, Algorithm {\sc
     Gamma-Solver} uses
   $p^{O(p)} \left( m^{\frac{p-2}{3p-2}} \log \left(
       \frac{m\|\bb\|_2^2}{\|\AA\|^2} \right)^{\frac{p}{3p-2}}
   \right)$ calls to a linear solver and returns a vector $\xx$ such
   that $\AA\xx = \bb$, and $\gamma_p(\tt,\xx)= p^{O(p)}O(1)$.
 \end{theorem}
 
We can now prove Theorem~\ref{thm:lpRegression}.

\begin{proof}
We will fix the value of $q = \max\{2, \sqrt{\log m} \}$.  If $p$ is
smaller than $q$, we can directly use the algorithm from
\cite{AdilKPS19}, which makes
$p^{O(p)} m^{\frac{p-2}{3p-2}} \poly(\log m) \log^{2} \frac{1}{\eps}
\le m^{\frac{p-2}{3p-2} + o(1)} \log^{2} \frac{1}{\eps}$ calls to a
linear solver.  We will look at the case when $p \geq q.$ Let us
assume we are starting from an $O(m)$-approximate solution to the
$p$-norm problem. We can assume this since we can use a homotopy
approach starting from $q$-norm solutions similar to Section~\ref{sec:p-to-q}.
The general iterative refinement scheme allows us to solve the residual
problem,
\[
\begin{aligned}
\max_{\Delta} \quad & \gg^{\top}\Delta - 2 \sum_e \rr_e \Delta_e^2 - \norm{\Delta}_p^p\\
& \AA\Delta = 0,
\end{aligned}
\]
at every iteration to a $\kappa$-approximation,
$O\left(p\kappa \log m \log\frac{1}{\eps} \right)$ times. We can now
do a binary search over the $O\left(\log\frac{pm}{\eps} \right)$
values $\nu$ of the residual problem and from Corollary \ref{cor:ApproxToRes}, we know it is sufficient to approximately
solve \eqref{eq:qProblemLargep}. Now we want to use Algorithm 4 from
\cite{AdilKPS19} for solving these smoothed $q$-problems. Note that
this algorithm solves for a slightly different objective, the $\gamma$
function defined above.
Using {\sc Gamma-Solver} for solving Problem~\eqref{eq:GammaScaled}, we get
 $\Dtil$ such that $\gamma_q(\hat{\tt},\Dtil) \leq q^{O(q)}
 O(1) = c$. From Lemma \ref{lem:Scaling}, we can get a $\Dhat$ such that
 it satisfies the constraints of \eqref{eq:gammaObj} and
 $\gamma_q(\tt,\Dhat) \leq q^{1+q/2}\nu c$. Now, from Lemma
 \ref{lem:Smoothed2Gamma}, program \eqref{eq:qProblemScaled} has
 objective at $\Dhat$ at most $2q^{2+q/2}\nu c$. We can now go back to
 problem \eqref{eq:qProblemLargep} by scaling $\Dhat$ appropriately to
 $\Dbar$, however the objective of \eqref{eq:qProblemLargep} at
 $\Dbar$ is the same as the objective of \eqref{eq:qProblemScaled} at $\Dhat$ 
 and therefore is at most $2q^{2+q/2}\nu c$. From Corollary
 \ref{cor:ApproxToRes} we can get an $O(2q^{2+q/2}c m^{\frac{k}{k-1}\left(\frac{1}{q}-\frac{1}{k}\right)}) \leq q^{O(q)} m^{\frac{1}{q}}$-approximate solution to the residual problem.
 
 We now have $\Delta$, an $q^{O(q)} m^{\frac{1}{q}}$-approximate
 solution to the residual problem. We require, $q^{O(q)} m^{\frac{q-2}{3q-2}+o(1)}$ calls to a linear solver to
 solve problem \eqref{eq:GammaScaled}, and we make
 $O\left(p m^{o(1)} \log\frac{1}{\eps} \right)$ calls to to solve
 \eqref{eq:GammaScaled}. Thus, the total number of iterations required
 to solve the $\ell_p$-norm problem is at most
 $p m^{\frac{p-2}{3p-2}+o(1)} \log^2\frac{1}{\eps}$ for
 $q = \sqrt{\log m}$.
\end{proof}

\section{Reduction to $\ell_p$-Constrained Problems}
\label{sec:BoxConstraint}

In this section, we will reduce the residual
problem~\eqref{eq:residual} to the following $\ell_p$-constrained
problem when the optimum of the residual problem lies between
$(\nu/2,\nu]$.
\begin{align}
\label{eq:lpBox}
\begin{aligned}
\max_{\Delta} \quad & \gg^{\top}\Delta - 2\sum_e \rr_e \Delta_e^2\\
\text{s.t.} \quad & \norm{\Delta}_p^p \leq \nu,\\
& \AA\Delta = 0.
\end{aligned}
\end{align}
Here $\gg$ and $\rr$ are as defined in the previous sections. We will
further reduce this problem to an $\ell_{\infty}$ constrained problem,
which is the above problem with the $\ell_p$ constraint replaced by an
$\ell_{\infty}$ constraint.  Variants of the $\ell_{\infty}$
constrained problem have been studied by Cohen
\etal.~\cite{CohenMTV17} in the context of matrix scaling and
balancing. The main advantage of the $\ell_{\infty}$ constrained
problem is that we do not have to compute $p$-th powers in the
objective. However, these computations are still required to compute
$\gg$ and $\rr$.

We first define our notion for approximation to $\ell_p$-constrained problems.
\begin{definition}[$(\alpha,\beta)$-Approximation]\label{def:approx_box}
  Let $\alpha, \beta \geq 1$. We say $\Dtil$ is an
  $(\alpha,\beta)$-approximation to problem \eqref{eq:lpBox} if
  $\AA\Dtil = 0$, $\norm{\Dtil}_p^p \leq \beta \nu$ and
  $\gg^{\top}\Dtil - \sum_e \rr_e \Dtil_e^2 \geq \frac{1}{\alpha}
  OPT$.
\end{definition}

%
We will next show that an $(\alpha,\beta)$-approximation to \eqref{eq:lpBox} gives an approximate solution to the residual problem.
\begin{restatable}{lemma}{lpBox}
\label{lem:lpBox}
Let $\nu$ be such that $\residual_p(\Dopt) \in (\nu/2,\nu]$. An $(\alpha,\beta)$-approximate solution to \eqref{eq:lpBox} gives a $16(\alpha^p\beta)^{1/(p-1)}$-approximation to the residual problem.
\end{restatable}

As an immediate corollary, we can replace the $\ell_p$-norm constraint
with an $\ell_{\infty}$-norm constraint, at the loss of a factor of 
$m^{\frac{p}{p-1}}$ in the approximation ratio.
\begin{corollary}
\label{cor:lpBox}
Let $\Dtil$ be a solution to 
\begin{align}
\label{eq:linftyBox}
\begin{aligned}
\max_{\Delta} \quad & \gg^{\top}\Delta - 2\sum_e \rr_e \Delta_e^2\\
& \norm{\Delta}_{\infty} \leq \nu^{1/p}\\
& \AA\Delta = 0.
\end{aligned}
\end{align}
 such that $\gg^{\top}\Dtil - 2\sum_e \rr_e \Dtil_e^2 \geq \frac{1}{\alpha}OPT$ and $\norm{\Delta}_{\infty} \leq \beta\nu^{1/p}$. Then $\Dtil$ is a $16(\alpha^p\beta^p m)^{1/(p-1)}$-approximate solution to the residual problem.
\end{corollary}
\begin{proof}
We know that $\|\Dtil\|_{\infty} \leq \beta\nu^{1/p}$. This implies that $\|\Dtil\|^p_p \leq m\|\Dtil\|_{\infty}^p \leq m \beta^p\nu$. Therefore, $\Dtil$ is an $(\alpha,m\beta^p)$ approximate solution to the $\ell_p$ constrained problem. From Lemma \ref{lem:lpBox} $\Dtil$ is a $16(\alpha^p\beta^p m)^{1/(p-1)}$ approximate solution to the residual problem.
\end{proof}

Thus, by solving an $\ell_{\infty}$ constrained problem to a constant
approximation, we can obtain an
$O\left(m^{1/(p-1)}\right)$-approximate solution to the residual
problem.


\section*{Acknowledgements}
DA is supported by SS's NSERC Discovery grant.  SS is supported by the
Natural Sciences and Engineering Research Council of Canada (NSERC), a
Connaught New Researcher award, and a Google Faculty Research award.
The authors would like to thank Richard Peng for several helpful
discussions, and anonymous reviewers for their useful suggestions.

\printbibliography
\appendix

\section{Proofs from Section \ref{sec:p-to-q}}
\ApproxLargep*

\begin{proof}
  From Lemma \ref{lem:qProblem}, we know that the optimum objective of
  \eqref{eq:qProblemLargep} is at most $\nu$. Since at $\Dtil$ the
  objective is at most $\beta \nu$,
  $2 \sum_e \rr_e \Dtil_e^2 \leq 2\cdot \beta\nu$ and,
  $\norm{\Dtil}_q^q \leq 2 \beta \nu^{q/k} m^{1-q/k}$,
  giving us,
  \[
    \norm{\Dtil}_{k}^{k} \leq (2\beta)^{k/q} m^{k/q- 1}\nu .
  \]
  Let
  $\Dbar = \frac{1}{16\beta}m^{-\frac{k}{k-1}\left(\frac{1}{q}- \frac{1}{k}\right)}\Dtil = \alpha \Dtil$.
  Now,
\[
  2\sum_e \rr_e \Dbar_e^2 = 2 \alpha^2 \sum_e \rr_e \Dtil_e^2 \leq
  \alpha \frac{1}{16\beta}\cdot 2 \beta \nu \leq \alpha \frac{\nu}{8}.
\]
Since $\frac{k}{q} - (k-1) \leq 0$ for $k \geq q \geq 2$,
\begin{equation}
  \norm{\Dbar}_{k}^{k} = \alpha^{k} \norm{\Dtil}_{k}^{k} \leq  \alpha  \frac{1}{(16\beta)^{k-1}}m^{-\left(\frac{k}{q}
        - 1\right)}\norm{\Dtil}_{k}^{k}  \leq \frac{\alpha}{8^{k-1}} 2^{\frac{k}{q} - (k-1)} \beta^{\frac{k}{q} - (k-1)} \nu \leq \alpha \frac{\nu}{8}.
\end{equation}
The above bounds imply,
\begin{align*}
2\sum_e \rr_e \Dbar_e^2 + \norm{\Dbar}_{k}^{k} & \leq  \alpha \frac{\nu}{4}.
\end{align*} 

\end{proof}

\ApproxToRes*
\begin{proof}
Let $\Dbar = \alpha \Dtil$. From Lemma \ref{lem:ApproxLargep} we know that,
\[
2\sum_e \rr_e \Dbar_e^2 + \norm{\Dbar}_{k}^{k}  \leq  \alpha \frac{\nu}{4}.
\]
Also,
\[
\gg^{\top} \Dbar  = \alpha \gg^{\top}\Dtil =  \alpha \frac{\nu}{2}.
\]
This gives us,
\begin{align*}
\gg^{\top} \Dbar -2\sum_e \rr_e \Dbar_e^2 - \norm{\Dbar}_{k}^{k} \geq  \alpha \frac{\nu}{4}
\geq \frac{1}{64\beta}m^{-\frac{k}{k-1}\left(\frac{1}{q}
        - \frac{1}{k}\right)}OPT.
\end{align*}

\end{proof}

\IterativeRefinement*
\begin{proof}
We will apply Lemma \ref{lem:OriginalIterativeRefinement} for $2p$-norms. The starting solution $\xx^{(0)}$ is an $O(1)$-approximate solution to the $p$-norm problem. We want to solve the $2p$-norm problem to an $O(1)$-approximation, i.e., $\eps = O(1)$. Let $\tilde{\xx}$ denote the optimum of the $p$-norm problem and $\xx^{\star}$ denote the optimum of the $2p$-norm problem.
\begin{equation}
\textstyle
\|\xx^{(0)}\|_{2p}^{2p} \leq \|\xx^{(0)}\|_{p}^{2p} \leq O(1) \|\tilde{\xx}\|_{p}^{2p}  \leq O(1) \|\xx^{\star}\|_{p}^{2p} \\ \leq O(1) m^{2p\left(\frac{1}{p}-\frac{1}{2p}\right)}\|\xx^{\star}\|_{2p}^{2p}.
\end{equation}
We thus have, $\|\xx^{(0)}\|_{2p}^{2p} - \|\xx^{\star}\|_{2p}^{2p} \leq O(m)\|\xx^{\star}\|_{2p}^{2p}$. Now applying Lemma \ref{lem:OriginalIterativeRefinement}, we get a total iteration count to be,
\[
O\left(p \kappa \log \left(\frac{\|\xx^{(0)}\|^{2p}_{2p} - OPT}{\eps OPT}\right)\right) \leq O\left(p \kappa \log m\right).
\]
\end{proof}

\BinarySearch*
\begin{proof}
Let $\xx^{\star}$ denote the optimum of the $r$-norm problem. We know that $\xx^{(0)}$ is an $O(1)$-approximate solution for the $k$-norm problem. 
\begin{equation}
\textstyle
\|\xx^{(0)}\|_{r}^{r} \leq \|\xx^{(0)}\|_{k}^{r} \leq O(1) \|\xx^{\star}\|_{k}^{r}  \leq O(1) m^{r\left(\frac{1}{k}-\frac{1}{r}\right)} \|\xx^{\star}\|_{r}^{r} = O(1) m^{\left(\frac{r}{k}-1\right)} \|\xx^{\star}\|_{r}^{r} 
\end{equation}
We know from the definition of the residual problem,
\[
\residual_r(\Dopt) \leq \|\xx^{(0)}\|_{r}^{r} - \|\xx^{\star}\|_{r}^{r} \leq \|\xx^{(0)}\|_{r}^{r} .
\]
Since our solution is not an $\alpha$-approximate solution, $\|\xx\|_{r}^{r} \geq \alpha \|\xx^{\star}\|_{r}^{r}$.
\begin{equation}
\textstyle
\residual_r(\Dopt) \geq \residual_r\left(\frac{\xx - \xx^{\star}}{16r}\right)\geq \frac{1}{16r}\left(\|\xx\|_{r}^{r} -\|\xx^{\star}\|_{r}^{r} \right)  \geq \frac{(\alpha -1)}{16r}\|\xx^{\star}\|_{r}^{r} \geq \frac{\Omega(1)(\alpha-1)}{r}m^{-\left(\frac{r}{k}-1\right)}\|\xx^{(0)}\|_{r}^{r}.
\end{equation}
We therefore have,
\[
\Omega(1)(\alpha-1)\frac{\|\xx^{(0)}\|_{r}^{r}}{rm^{\left(\frac{r}{k}-1\right)}} \leq \residual_r(\Dopt) \leq \|\xx^{(0)}\|_{r}^{r}.
\]
\end{proof}

The following is a version of Lemma A.3 from \cite{AdilPS19}.
\begin{lemma}
\label{lem:Decision}
Let $\nu$ be such that the residual problem for $k$-norms satisfies $\residual_k(\Dopt) \in (\nu/2,\nu]$. The following problem has optimum at most $\nu$.
\begin{align}
\label{eq:Decision}
\begin{aligned}
\min_{\Delta \in \mathbb{R}^m} \quad& 2\sum_e \rr_e\Delta_e^2 +\norm{\Delta}_{k}^{k}\\
& \gg^{\top} \Delta = \nu/2\\
& \AA\Delta = 0.
\end{aligned}
\end{align}
\end{lemma}

\begin{lemma}
\label{lem:qProblem}
Let $\nu$ be such that the residual problem for $k$-norms satisfies $\residual_k(\Dopt) \in (\nu/2,\nu]$. Problem \eqref{eq:qProblemLargep} has optimum at most $\nu$ when $k\geq q$.
\begin{proof}
Let $\Dopt$ be the optimum of problem \eqref{eq:Decision}. From Lemma \ref{lem:Decision}, we know that $\norm{\Dopt}_{k}^{k} \leq \nu$ and 
$2 \sum_e \rr_e\Dopt_e^2 \leq \nu$. Now,
\[
 \norm{\Dopt}_{q}^{q} \leq  m^{q\left(\frac{1}{q} - \frac{1}{k}\right)}\norm{\Dopt}_{k}^{q} \leq m^{1-\frac{q}{k}} \nu^{q/k}.
\]
The bound now follows from noting,
\[
 \sum_e \rr_e\Dopt_e^2 + \frac{\nu^{1 - q/k}}{2}m^{-(1-q/k)} \norm{\Dopt}_{q}^{q} \leq \frac{\nu}{2} +\frac{\nu}{2}.
\]
\end{proof}
\end{lemma}

\IterativeRefinementSmall*

\begin{proof}
We will apply Lemma \ref{lem:OriginalIterativeRefinement} for $p$-norms. The starting solution $\xx^{(0)}$ is an $O(1)$-approximate solution to the $q$-norm problem. We want to solve the $p$-norm problem to an $O(1)$-approximation, i.e., $\eps = O(1)$. Let $\tilde{\xx}$ denote the optimum of the $q$-norm problem and $\xx^{\star}$ denote the optimum of the $p$-norm problem.
\begin{multline}
\textstyle
\|\xx^{(0)}\|_{p}^{p} \leq m^{p\left(\frac{1}{p}-\frac{1}{q}\right)} \|\xx^{(0)}\|_{q}^{p} \leq O(1)m^{p\left(\frac{1}{p}-\frac{1}{q}\right)}  \|\tilde{\xx}\|_{q}^{p}  \\ \leq O(1)m^{p\left(\frac{1}{p}-\frac{1}{q}\right)}  \|\xx^{\star}\|_{q}^{p} \leq O(1) m^{p\left(\frac{1}{p}-\frac{1}{q}\right)} \|\xx^{\star}\|_{p}^{p}.
\end{multline}
We thus have, $\|\xx^{(0)}\|_{p}^{p} - \|\xx^{\star}\|_{p}^{p} \leq O(m)\|\xx^{\star}\|_{p}^{p}$. Now applying Lemma \ref{lem:OriginalIterativeRefinement}, we get a total iteration count to be,
\[
O\left(p \kappa \log \left(\frac{\|\xx^{(0)}\|^{p}_{p} - OPT}{\eps OPT}\right)\right) \leq O\left(p \kappa \log m\right).
\]
\end{proof}

\BinarySearchSmall*
\begin{proof}
Let $\xx^{\star}$ denote the optimum of the $p$-norm problem. We know that $\xx^{(0)}$ is an $O(1)$-approximate solution for the $q$-norm problem. 
\[
m^{-(1-p/q)}\|\xx^{(0)}\|_{p}^{p} \leq \|\xx^{(0)}\|_{q}^{p} \leq  \|\xx^{\star}\|_{q}^{p} \leq  \|\xx^{\star}\|_{p}^{p}  
\]
We know from the definition of the residual problem,
\[
\residual_p(\Dopt) \leq \|\xx^{(0)}\|_{p}^{p} - \|\xx^{\star}\|_{p}^{p} \leq \|\xx^{(0)}\|_{p}^{p} .
\]
Since our solution is not an $\alpha$-approximate solution, $\|\xx^{(0)}\|_{p}^{p} \geq \alpha \|\xx^{\star}\|_{p}^{p}$.
\begin{align*}
\textstyle
\residual_p(\Dopt) & \geq \residual_p\left(\frac{\xx - \xx^{\star}}{16p}\right)\geq  \frac{1}{16p}\left(\|\xx^{(0)}\|_{p}^{p} -\|\xx^{\star}\|_{p}^{p} \right) \geq \frac{(\alpha -1)}{16p}\|\xx^{\star}\|_{p}^{p} \\ &\geq \frac{\Omega(1)(\alpha-1)}{p}m^{-(1-p/q)}\|\xx^{(0)}\|_{p}^{p}.
\end{align*}
We therefore have,
\[
 \Omega(1)(\alpha-1)\frac{\|\xx^{(0)}\|_{p}^{p}}{p m} \leq \residual_p(\Dopt) \leq \|\xx^{(0)}\|_{p}^{p}.
\]
\end{proof}

\begin{lemma}
\label{lem:qProblemSmallp}
Let $\nu$ be such that the residual problem for $k$-norms satisfies $\residual_k(\Dopt) \in (\nu/2,\nu]$. Problem \eqref{eq:qProblemSmallp} has optimum at most $\nu$ when $k<q$.
\end{lemma}
\begin{proof}
Let $\Dopt$ be the optimum of problem \eqref{eq:Decision}. From Lemma \ref{lem:Decision}, we know that $\norm{\Dopt}_{k}^{k} \leq \nu$ and 
$2 \sum_e \rr_e\Dopt_e^2 \leq \nu$. Now,
\[
\norm{\Dopt}_{q}^{q} \leq  \norm{\Dopt}_{k}^{q} \leq \nu^{q/k}.
\]
The bound now follows from noting,
\[
\sum_e \rr_e\Dopt_e^2 + \frac{\nu^{1 - q/k}}{2^{q/k}} \norm{\Dopt}_{q}^{q} \leq \frac{\nu}{2} +\frac{\nu}{2}.
\]
\end{proof}

\ApproxSmallp*

\begin{proof}
From Lemma \ref{lem:qProblemSmallp}, we know that the objective of \eqref{eq:qProblemSmallp} is at most $\nu$. Since we have a $\beta$-approximate solution, $2 \sum_e \rr_e \Dtil_e^2 \leq 2\cdot \beta\nu$ and, $\norm{\Dtil}_q^q \leq 2^{q/k} \beta \nu^{q/k}$ and,
\[
\norm{\Dtil}_{k}^{k} \leq 2 m^{1-k/q} \beta^{k/q}\nu  \leq 2 m^{1-k/q} \beta \nu.
\]
Let
$\Dbar =\frac{1}{16\beta }m^{-\frac{k}{k-1}\left(\frac{1}{k}
      - \frac{1}{q}\right)} \Dtil = \alpha \Dtil$. Now,
      

\[
2\sum_e \rr_e \Dbar_e^2 =  2\alpha^2  \sum_e \rr_e \Dtil_e^2 \leq \alpha \frac{\nu}{8}.
\]
and,
\begin{equation}
\norm{\Dbar}_{k}^{k} = \alpha^{k} \norm{\Dtil}_{k}^{k}  \leq  \alpha  \frac{1}{(16\beta)^{k-1}}m^{-k\left(\frac{1}{k} - \frac{1}{q}\right)} \cdot 2 m^{1-k/q}\beta\nu  = \alpha  \frac{\nu}{8}.
\end{equation}
The above bounds imply,
\begin{align*}
2\sum_e \rr_e \Dbar_e^2 + \norm{\Dbar}_{k}^{k} & \leq  \alpha \frac{\nu}{4}.
\end{align*} 
Also,
\[
\gg^{\top} \Dbar  = \alpha \gg^{\top}\Dtil =  \alpha \frac{\nu}{2}.
\]
We now get, 
\begin{align*}
\gg^{\top} \Dbar -2\sum_e \rr_e \Dbar_e^2 - \norm{\Dbar}_{k}^{k} \geq  \alpha \frac{\nu}{4} 
\geq \frac{1}{64\beta}m^{-\frac{k}{k-1}\left(\frac{1}{k} - \frac{1}{q}\right)} OPT.
\end{align*}
\end{proof}

\section{Proofs from Section~\ref{sec:unweighted-flow}}
\qSTOCoracle*

\begin{proof}
  We have that the residual $p$-norm problem has a value in
  $(\nu/2, \nu].$ Consider the optimal solution $\Delta^{\star}$ to
  the following residual $p$-norm problem:
  \[\max_{\Delta: \AA\Delta = \vzero} \gg^{\top} \Delta - 2 \sum_{e}
    \rr_e \Delta_e^{2} - \norm{\Delta}_{p}^{p}.\] Consider the
  solutions $\lambda\Delta^{*}.$ Since all these solutions are
  feasible as we have $\AA(\lambda \Delta^{\star}) =0,$ we know that
  the objective is optimal for $\lambda = 1.$ Thus, differentiating
  with respect to $\lambda$ at $\lambda=1$ gives,
  \[\gg^{\top} \Delta^{\star} - 4 \sum_{e}
    \rr_e \left( \Delta^{\star}_e \right)^{2} - p \norm{\Delta^{\star}}_{p}^{p} = 0. \]
  Rearranging
  \begin{equation} 2 \sum_{e} \rr_e \left( \Delta^{\star}_e \right)^{2} + (p-1)
    \norm{\Delta^{\star}}_{p}^{p} \\ = \gg^{\top} \Delta^{\star} - 2
    \sum_{e} \rr_e \left( \Delta^{\star}_e \right)^{2} -
    \norm{\Delta^{\star}}_{p}^{p} \le \nu.
  \end{equation}
  Since $p \ge 2,$ we get
  $\norm{\Delta^{\star}}_{p} \le \nu^{\nfrac{1}{p}}.$ Thus,
  $\norm{\nu}_q \le m^{\frac{1}{q} - \frac{1}{p}}\nu^{\nfrac{1}{p}}$.

  Consider the problem \eqref{eq:qSTOC}.
  First observe that this problem is of the form that can be solved
  using Theorem~\ref{thm:qOracle}. Moreover, considering
  $-\Delta^{\star}$ as a feasible solution, we have that the objective
  is at most 
  \begin{align*}
    - \gg^{\top} \Delta^{\star} + 2 \sum_{e}
    \rr_e \left( \Delta_e^{\star} \right)^{2} + \frac{1}{4} \left( \frac{\nu}{m}
    \right)^{1-\frac{q}{k}} \norm{\Delta^{\star}}_{q}^{q}
    \le \frac{-\nu}{2} + \frac{\nu}{4} = \frac{-\nu}{4}.
  \end{align*}

  Now, suppose we are given a solution $\Dtil$ to the above smoothed
  $q$-norm problem with objective value at most $-\frac{\nu}{16}.$ We
  will show that a scaling of $\Dtil$ provides a good solution to the
  residual $p$-norm problem.

First, we assume, $\abs{\gg^{\top} \Dtil} \le \nu.$ Since $\Dtil$ has objective at most $-\frac{\nu}{16},$ we must have,
  \[ 2 \sum_{e}
    \rr_e \Dtil_e^{2} + \frac{1}{4} \left( \frac{\nu}{m} \right)^{1-\frac{q}{p}}
    \norm{\Dtil}_{q}^{q} \le -\frac{\nu}{16} + \nu \le \nu. \]

  Thus,
  $\norm{\Dtil}_{q} \le 4^{\frac{1}{q}} \nu^{\frac{1}{p}}
  m^{\frac{1}{q} - \frac{1}{p}},$ and hence
  $\norm{\Dtil}_{p}^{p} \le 4^{\frac{p}{q}} \nu m^{\frac{p}{q} - 1}.$

Let $\Dbar = -\alpha \Dtil,$ where
  $\alpha = \frac{1}{256} m^{-\frac{p}{p-1}\left(\frac{1}{q}-
      \frac{1}{p} \right)}$.  We show that $\Dbar$ provides a good
  solution to the residual $p$-norm problem. 
  Hence, the objective of the $p$-norm residual problem becomes
  \begin{align*}
 &   -\alpha \gg^{\top} \Dtil - 2 \alpha^{2} \sum_{e}
    \rr_e \Dtil_e^{2} - \alpha^{p} \norm{\Dtil}_{p}^{p}\\
    & \ge \frac{\alpha \nu}{16} - \frac{\alpha}{256} \cdot 2 \sum_{e}
      \rr_e \Dtil_e^{2} - \alpha \alpha^{p-1} 4^{\frac{p}{q}} \nu
      m^{\frac{p}{q}-1} \\
    & \ge  \frac{\alpha \nu}{16} - \frac{\alpha \nu}{256} -
      \frac{\alpha \nu}{64} \ge \frac{\alpha\nu}{64}.
  \end{align*}

  For the case $\abs{\gg^{\top} \Dtil} \ge \nu,$ consider the vector
  $z\Dtil,$ where
  $z=\frac{\nu}{2\abs{\gg^{\top} \Dtil}} \le \frac{1}{2}.$ Observe
  that this vector is still feasible for the smoothed $q$-norm problem
  given by Program~\eqref{eq:qSTOC}. Moreover, we have
  $ \gg^{\top}(z\Dtil) = -\frac{\nu}{2},$ and its objective for the
  same program is
  \begin{align*}
    z\gg^{\top} \Dtil + 2z^{2}  \sum_{e}
    \rr_e \Delta_e^{2} +     z^{q} \frac{1}{4} \left( \frac{\nu}{m}
    \right)^{1-\frac{q}{p}} \norm{\Delta}_{q}^{q}
     \le -\frac{\nu}{2} +
    z^{2} \nu \le -\frac{\nu}{4}.
  \end{align*}
  Thus, we can repeat the argument for the case $\abs{\gg^{\top}
    \Dtil} \le \nu$ to obtain our vector.
\end{proof}

\begin{restatable}{lemma}{qSTOCoracleSmallp}
\label{lem:appendix:qSTOCoracleSmallp}
Let $q\geq p \geq 2$ and $\nu$ be such that
$\residual_p(\Dopt) \in (\nu/2,\nu]$, where $\Dopt$ is the optimum of
the residual problem for $q$-norm. The following problem has optimum
at most $-\frac{\nu}{4}.$
\begin{align}
  \label{eq:qSTOCSmallp}
\begin{aligned}
  \min_{\Delta: \AA\Delta = \vzero} \gg^{\top} \Delta + 2 \sum_{e}
  \rr_e \Delta_e^{2} + \frac{\nu^{1-\frac{q}{p}}}{4}   \norm{\Delta}_{q}^{q}.    
\end{aligned}
\end{align}
If $\Dtil$ is a feasible solution to the above program such that the
objective is at most $-\frac{\nu}{16}$, then a scaling of $\Dtil$
gives us a feasible solution to $\residual_p$ with objective value
$\Omega(\nu m^{-\frac{p}{p-1}\left(\frac{1}{p}- \frac{1}{q} \right)}).$
\end{restatable}

\begin{proof}
  We have that the residual $p$-norm problem has a value in
  $(\nu/2, \nu].$ Consider the optimal solution $\Delta^{\star}$ to
  the following residual $p$-norm problem:
  \[\max_{\Delta: \AA\Delta = \vzero} \gg^{\top} \Delta - 2 \sum_{e}
    \rr_e \Delta_e^{2} - \norm{\Delta}_{p}^{p}.\] Consider the
  solutions $\lambda\Delta^{*}.$ Since all these solutions are
  feasible as we have $\AA(\lambda \Delta^{\star}) =0,$ we know that
  the objective is optimal for $\lambda = 1.$ Thus, differentiating
  with respect to $\lambda$ at $\lambda=1$ gives,
  \[\gg^{\top} \Delta^{\star} - 4 \sum_{e}
    \rr_e \left( \Delta^{\star}_e \right)^{2} - p \norm{\Delta^{\star}}_{p}^{p} = 0. \]
  Rearranging
  \begin{equation} 2 \sum_{e} \rr_e \left( \Delta^{\star}_e \right)^{2} + (p-1)
    \norm{\Delta^{\star}}_{p}^{p} \\ = \gg^{\top} \Delta^{\star} - 2
    \sum_{e} \rr_e \left( \Delta^{\star}_e \right)^{2} -
    \norm{\Delta^{\star}}_{p}^{p} \le \nu.
  \end{equation}
  Since $p \ge 2,$ we get
  $\norm{\Delta^{\star}}_{p} \le \nu^{\nfrac{1}{p}}.$ Thus,
  $\norm{\Dopt}_q \le \nu^{\nfrac{1}{p}}$

  Consider the problem \eqref{eq:qSTOCSmallp}.
  First observe that this problem is of the form that can be solved
  using Theorem~\ref{thm:qOracle}. Moreover, considering
  $-\Delta^{\star}$ as a feasible solution, we have that the objective
  is at most
  \begin{align*}
    - \gg^{\top} \Delta^{\star} + 2\sum_{e}
    \rr_e \left( \Delta_e^{\star} \right)^{2} +  \frac{1}{4}  \nu^{1-\frac{q}{p}} \norm{\Delta^{\star}}_{q}^{q}
    \le \frac{-\nu}{2} + \frac{\nu}{4} = \frac{-\nu}{4}.
  \end{align*}

  Now, suppose we are given a solution $\Dtil$ to the above smoothed
  $q$-norm problem with objective value at most $-\frac{\nu}{16}.$ We
  will show that a scaling of $\Dtil$ provides a good solution to the
  residual $p$-norm problem.

    First, we assume, $\abs{\gg^{\top} \Dtil} \le \nu.$ Since $\Dtil$ has objective at most $-\frac{\nu}{16},$ we must have,
  \[ 2 \sum_{e}
    \rr_e \Dtil_e^{2} + \frac{1}{4}  \nu^{1-\frac{q}{p}}
    \norm{\Dtil}_{q}^{q} \le -\frac{\nu}{16} + \nu \le \nu. \]

  Thus,
  $\norm{\Dtil}_{q} \le 4^{\frac{1}{q}} \nu^{\frac{1}{p}},$ and hence
  $\norm{\Dtil}_{p}^{p} \le 4^{\frac{p}{q}} \nu m^{1-\frac{p}{q} } \leq 4\nu m^{1-\frac{p}{q} } .$

 Let $\Dbar = -\alpha \Dtil,$ where
  $\alpha = \frac{1}{256} m^{-\frac{p}{p-1}\left(\frac{1}{p}-
      \frac{1}{q} \right)}$.  We show that $\Dbar$ provides a good
  solution to the residual $p$-norm problem. 
  Hence, the objective of the $p$-norm residual problem becomes
  \begin{align*}
    &-\alpha \gg^{\top} \Dtil - 2 \alpha^{2} \sum_{e}
    \rr_e \Dtil_e^{2} - \alpha^{p} \norm{\Dtil}_{p}^{p}\\
    & \ge \frac{\alpha \nu}{16} - \frac{\alpha}{256} \cdot 2 \sum_{e}
      \rr_e \Dtil_e^{2} - \alpha \alpha^{p-1} 4 \nu
      m^{1-\frac{p}{q}} \\
    & \ge  \frac{\alpha \nu}{16} - \frac{\alpha \nu}{256} -
      \frac{\alpha \nu}{64} \ge \frac{\alpha\nu}{64}.
  \end{align*}

  For the case $\abs{\gg^{\top} \Dtil} \ge \nu,$ consider the vector
  $z\Dtil,$ where
  $z=\frac{\nu}{2\abs{\gg^{\top} \Dtil}} \le \frac{1}{2}.$ Observe
  that this vector is still feasible for the smoothed $q$-norm problem
  given by Program~\eqref{eq:qSTOC}. Moreover, we have
  $ \gg^{\top}(z\Dtil) = -\frac{\nu}{2},$ and its objective for the
  same program is
  \begin{align*}
    z\gg^{\top} \Dtil + 2z^{2}  \sum_{e}
    \rr_e \Delta_e^{2} +     z^{q} \frac{1}{4} \nu^{1-\frac{q}{p}} \norm{\Delta}_{q}^{q}\le -\frac{\nu}{2} +
    z^{2} \nu \le -\frac{\nu}{4}.
  \end{align*}
  Thus, we can repeat the argument for the case $\abs{\gg^{\top}
    \Dtil} \le \nu$ to obtain our vector.
\end{proof}

\section{Proofs from Section \ref{sec:regression}}

Define the following function which is the sum of the quadratic term and the $q$-norm term from the residual problem,
\[
  h_q(\rr,\Delta) = 2 \sum_e \rr_e\Delta_e^2 + 
  \norm{\Delta}_q^q.
\]
 The following lemma relates the functions $h$ and $\gamma$ for any $q \geq 2$.
 \begin{restatable}{lemma}{resgamma}
   \label{lem:res-gamma}
   Let $h_q(\rr,\Delta)$ and $\gamma_q(\tt,\Delta)$ be as defined above. The following holds for any $\Delta$, any $\tt$ and $\rr$ such that $\rr = \tt^{p-2}$, and
   $q\geq 2$.
 \[
  \frac{1}{q} \gamma_q(\tt,\Delta) \leq h_q(\rr,\Delta) \leq 3
   \gamma_q(\tt,\Delta).
 \]
 \end{restatable}

\begin{proof}
 We will only show the above relation for one coordinate. Let us look at the two cases, 
 \begin{enumerate}
 \item $|\Delta| \leq t:$ We want to show,
 \[
\frac{1}{q}\left( \frac{q}{2}t^{q-2}\Delta^2\right) \leq 2t^{q-2}\Delta^2  + |\Delta|^q \leq 3 \left(\frac{q}{2}t^{q-2}\Delta^2\right).
 \]
 The left inequality directly follows. For the other side,
 \[
 2 t^{q-2}\Delta^2  + |\Delta|^q \leq  3 t^{q-2}\Delta^2 \leq 3 \left(\frac{q}{2}t^{q-2}\Delta^2\right).
 \]
\item $|\Delta| \geq t:$ We want to show,
\begin{equation}
\textstyle
\frac{1}{q} \left(|\Delta|^q +\left(\frac{q}{2}-1\right) t^q\right) \leq 2 t^{q-2}\Delta^2  +  |\Delta|^q \\ \leq 3\left(|\Delta|^q +\left(\frac{q}{2}-1\right) t^q\right).
\end{equation}
To see the left inequality note that, $|\Delta|^q +\left(\frac{q}{2}-1\right) t^q \leq \frac{q}{2}|\Delta|^q$ and the rest follows. For the right inequality, 
\[
2 t^{q-2}\Delta^2  + |\Delta|^q \leq 3 |\Delta|^q,
\]
and the rest follows.
 \end{enumerate}
 \end{proof}

\SmoothedGamma*

 \begin{proof}
 Let $\rr' = \nu^{-2\left(\frac{1}{q}-\frac{1}{p}\right)}m^{2\left(\frac{1}{q}-\frac{1}{p}\right)} \rr$. Then $\tt = \rr'^{1/(q-2)}$. The objective of \eqref{eq:qProblemScaled} is now, $\sum_e \rr'_e \zeta_e^2 + \frac{1}{2}\norm{\zeta}_{q}^{q}.$
 Note that,
 \[
 \sum_e \rr'_e \zeta_e^2 + \frac{1}{2}\norm{\zeta}_{q}^{q} \leq  2\sum_e \rr'_e \zeta_e^2 + \norm{\zeta}_{q}^{q}  =  h_q(\rr',\zeta),
 \]
 and,
 \begin{equation}
 \textstyle
 \sum_e \rr'_e \zeta_e^2 + \frac{1}{2}\norm{\zeta}_{q}^{q}\geq \frac{1}{2}\left( 2\sum_e \rr'_e \zeta_e^2 + \norm{\zeta}_{q}^{q} \right) \\ = \frac{1}{2} h_q(\rr',\zeta).
\end{equation}
 Let us denote the objective of \eqref{eq:qProblemScaled} as a function of $\zeta$ as $s(\zeta)$. From the above inequalities and Lemma \ref{lem:res-gamma}, we have
 \[
\frac{1}{2q}\gamma_q(\tt,\zeta) \leq s(\zeta) \leq 3\gamma_q(\tt,\zeta).
 \]
 Now, since \eqref{eq:qProblemScaled}is only a scaling of problem \eqref{eq:qProblemLargep}, they have the same value of optimum objective. Therefore optimum of \eqref{eq:qProblemScaled} is at most $\nu$. From the above relation, we know that $\gamma_q(|\xx|,\zeta) \leq 2q s(\zeta)$ and therefore, the optimum of \eqref{eq:gammaObj} is at most $2q \nu$.
%
 \end{proof}
 
 \scaling*
\begin{proof}
Suppose $\Delta$ is the optimum of \eqref{eq:gammaObj}.We know that $\gamma_q(\tt,\Delta) \leq 2q \nu$ and $\gg^{\top}\Delta =  m^{-\left(\frac{1}{q}-\frac{1}{p}\right)} \nu^{1+\frac{1}{q}-\frac{1}{p}}/2$
Scaling both $\tt$ and $\Delta$ to $\tilde{\tt} =  (2q \nu)^{-1/q}\tt$ and $\Dtil = (2q \nu)^{-1/q} \Delta$ gives the following.
\[\begin{aligned}
\textstyle
       & \gamma_q(\tilde{\tt},\Dtil) \leq 1\\
       & \gg^{\top}\Dtil = (2q)^{-1/q}p m^{-\left(\frac{1}{q}-\frac{1}{p}\right)} \nu^{1-\frac{1}{p}}/2\\
       & \AA\Dtil = 0.
     \end{aligned}\]
 Now, let $\tt' = \max\{m^{-1/q},\tt\}$. We claim that $\gamma_q (\tt',\Dtil) - \gamma_q(\tilde{\tt},\Dtil) \leq \frac{q}{2}-1$. To see this, for a single $j$, let us look at the difference $\gamma_q (\tt'_j,\Dtil_j) - \gamma_q(\tilde{\tt}_j,\Dtil_j)$. If $\tilde{\tt}_j \geq m^{-1/q}$ the difference is $0$. Otherwise, from the proof of Lemma $5$ of \cite{BubeckCLL18},
 \begin{equation}
 \textstyle
 \gamma_q (\tt'_j,\Dtil_j) - \gamma_q(\tilde{\tt}_j,\Dtil_j) \leq \gamma_q (\tt'_j,\Dtil_j) - |\Dtil_j|^q  \leq \left(\frac{q}{2}-1\right)(m^{-1/q})^q.
 \end{equation}
We know that, $\gamma_q(\tilde{\tt},\Dtil) \leq 1$. Thus, $\gamma_q(\tt',\Dtil) \leq \frac{q}{2}$. Next we set, $\Dhat = \left(\frac{2}{q}\right)^{1/2}\Dtil$. Now, $\gamma_q(\tt',\Dhat) \leq \frac{2}{q} \gamma_q(\tt',\Dtil) \leq 1$. Define $\hat{\tt} = \min \{1,\tt'\}$. Note that $\gamma_q(\hat{\tt},\Dhat) \leq \gamma_q(\tt',\Dhat) \leq 1$ since the optimum is at most $1$. Suppose $\Dopt$ is a $\kappa$-approximate solution of \eqref{eq:GammaScaled}.
\[
\gamma_q(\hat{\tt},\Dopt) \leq \kappa \cdot OPT \leq \kappa.
\] 
$\gamma_q$ is an increasing function of $\tt$ since $q \geq 2$. This gives us,
\[
\gamma_q(\tilde{\tt},\Dopt) \leq \gamma_q(\tt',\Dopt) = \gamma_q(\hat{\tt},\Dopt) \leq \kappa.
\]
This gives,
\[
\gamma_q(\tt,(2q \nu)^{1/q}\Dopt) \leq 2q \nu \kappa.
\]
Finally,
\[ 
\textstyle
\gamma_q\left(\tt,\left(\frac{q}{2}\right)^{1/2}(2q \nu)^{1/q}\Dopt\right) \leq \left(\frac{q}{2}\right)^{q/2}2q \nu \kappa \leq q^{1+q/2}\nu\kappa.
\]
\end{proof}

\section{Proofs from Section \ref{sec:BoxConstraint}}
\lpBox*
\begin{proof}
Let $\Dopt$ denote the optimum of the residual problem. Since $\norm{\Dopt}_p^p \geq 0$, we can conclude that $\gg^{\top}\Dopt - 2\sum_e \rr_e \Dopt_e^2 \geq \nu/2$. At the optimum,
\[
\frac{d}{d\lambda}\left[\gg^{\top}(\lambda\Dopt) - 2\sum_e \rr_e (\lambda\Dopt_e)^2 - \|\lambda\Dopt\|_p^p\right]_{\lambda = 1} = 0.
\]
This implies,
\begin{equation}
 2\sum_e \rr_e \Dopt_e^2 + (p-1)\|\Dopt\|_p^p \\= \gg^{\top}\Dopt - 2\sum_e \rr_e \Dopt_e^2 - \|\Dopt\|_p^p \leq \nu.
\end{equation}
We thus have $\norm{\Dopt}_p^p \leq \nu$ which is a feasible solution for \eqref{eq:lpBox}. We can thus conclude that \eqref{eq:lpBox} has an optimum at least $\nu/2$. Let $\Dtil$ denote an $(\alpha,\beta)$-approximate solution to \eqref{eq:lpBox}. We know that,
\[
\gg^{\top}\Dtil - 2\sum_e \rr_e \Dtil_e^2 \geq \frac{1}{\alpha} \frac{\nu}{2},
\]
and,
\[
\norm{\Dtil}_p^p \leq \beta \nu.
\] 
Let $\Delta = \frac{1}{(4\alpha\beta)^{1/(p-1)}} \Dtil$. Now, $
\norm{\Delta}_p^p \leq \frac{1}{(4\alpha\beta)^{1/(p-1)}} \frac{1}{\alpha} \frac{\nu}{4}
$
and,
\begin{align*}
&\gg^{\top}\Delta - 2\sum_e \rr_e \Delta_e^2 \\
& =  \frac{1}{(4\alpha\beta)^{1/(p-1)}}  \left(\gg^{\top}\Dtil - 2\frac{1}{(4\alpha\beta)^{1/(p-1)}}\sum_e \rr_e \Dtil_e^2\right) \\
&\geq \frac{1}{(4\alpha\beta)^{1/(p-1)}}\left(\gg^{\top}\Dtil - 2\sum_e \rr_e \Dtil_e^2\right)\\
& \geq \frac{1}{(4\alpha\beta)^{1/(p-1)}} \frac{1}{\alpha} \frac{\nu}{2}.
\end{align*}
From the above calculations, we can conclude that,
\begin{equation}
\gg^{\top}\Delta - 2 \sum_e \rr_e \Delta_e^2 - \norm{\Delta}_p^p  \geq \frac{1}{(4\alpha\beta)^{1/(p-1)}} \frac{1}{\alpha} \frac{\nu}{4} \\  \geq \frac{1}{16(\alpha^{p}\beta)^{1/(p-1)}} OPT .
\end{equation}
\end{proof}


\end{document}